\pdfoutput=1
\documentclass[letterpaper, reqno,11pt]{article}
\usepackage[margin=1.25in]{geometry}
\usepackage{mathpazo}
\usepackage{mathtools,amsthm, latexsym, color,graphicx, xspace, url}
\usepackage{amsfonts}
\usepackage{euscript}
\usepackage{calrsfs}
\usepackage{bbold}
\usepackage{cite}
\usepackage{microtype}
\usepackage{xcolor}

\usepackage{fixltx2e}
\usepackage{fullpage}
\usepackage{xspace}
\usepackage{tikz}  %
\usepackage{hyperref}%
\hypersetup{%
   breaklinks,%
   ocgcolorlinks,
   colorlinks=true,%
   linkcolor=[rgb]{0.0,0.0,0.0},%
   citecolor=[rgb]{0,0,0.0}
}
\makeatletter
\newcommand{\Xcite}[1]{{\let\cite@adjust\empty\cite{#1}}}
\makeatother

\theoremstyle{plain}

\newtheorem{theorem}{Theorem}[section]
\newtheorem{lemma}[theorem]{Lemma}
\newtheorem{proposition}[theorem]{Proposition}
\newtheorem{cor}[theorem]{Corollary}

\newcommand{\vardim}{g}

\newcommand{\RR}{\mathbb{R}}

\newcommand{\poly}{\operatorname{Poly}}
\newcommand{\sign}{\operatorname{sign}}
\newcommand{\XX}{\mathbb{X}}
\newcommand{\YY}{\mathbb{Y}}

\newcommand{\Q}{\mathcal{Q}}

\newcommand{\PP}{\mathcal{P}}
\let\eps\varepsilon
\let\bd\partial
\def\DS{\mathsf{DS}}
\def\column#1{{{\uparrow}#1}}
\def\shadow#1{{{\downarrow}#1}}

\def\A{\mathcal{A}}
\def\C{\mathcal{C}}
\def\G{\mathcal{G}}
\def\K{\mathcal{K}}
\def\V{\mathcal{V}}
\def\P{\mathcal{P}}
\def\Q{\mathcal{Q}}
\def\R{\mathcal{R}}

\def\S{\EuScript{S}}
\def\T{\mathcal{T}}

\def\ZZ{\mathsf{Z}}

\def\etal{\textit{et~al.}}

\DeclareMathOperator{\polylog}{polylog}
\DeclareMathOperator{\St}{St}
\DeclarePairedDelimiter\abs{\lvert}{\rvert}

\begin{document}

\title{An Efficient Algorithm for Generalized Polynomial Partitioning and Its Applications\thanks{
    P.~Agarwal was supported by NSF under grants CCF-15-13816, CCF-15-46392, and IIS-14-08846, by an ARO grant W911NF-15-1-0408, and by BSF Grant 2012/229 from the U.S.-Israel Binational Science Foundation.
    B.~Aronov was supported by NSF grants CCF-12-18791 and CCF-15-40656, and by grant 2014/170 from the US-Israel Binational Science Foundation.
    E.~Ezra was supported by NSF CAREER grant CCF:AF 1553354 and by Grant 824/17 from the Israel Science Foundation.
    J.~Zahl was supported by a NSERC Discovery grant.
    A preliminary version of this work appeared in the \emph{Proceedings of the 35th International Symposium on Computational Geometry} \cite{AAEZ}.}
}
\author{
  Pankaj K. Agarwal%
  \footnote{Department of Computer Science, Duke University, Box 90129, Durham, NC 27708-0129 USA;
    \texttt{pankaj@cs.duke.edu}.
  }
  \and
  Boris Aronov%
  \footnote{Department of Computer Science and Engineering,
    Tandon School of Engineering, New York University, Brooklyn, NY~11201, USA;
    \texttt{boris.aronov@nyu.edu}.
  }
  \and
  Esther Ezra%
  \footnote{Department of Computer Science, Bar-Ilan University, Ramat Gan, Israel  and
    School of Math, Georgia Institute of Technology, Atlanta, Georgia 30332, USA;
    \texttt{ezraest@cs.biu.ac.il}.
  }
  \and
  Joshua Zahl\footnote{Department of Mathematics, University of British Columbia, Vancouver, BC~V6T 1Z2, Canada;
    \texttt{jzahl@math.ubc.ca}.}
}

\maketitle

\begin{abstract}
  In 2015, Guth proved that if $\S$ is a collection of $n$ $\vardim$-dimensional semi-algebraic sets in ${\RR}^d$ and if $D\geq 1$ is an integer, then there is a $d$-variate polynomial $P$ of degree at most $D$ so that each connected component of $\RR^d\setminus Z(P)$ intersects $O(n/D^{d-\vardim})$ sets from $\S$. Such a polynomial is called a \emph{generalized partitioning polynomial}. We present a randomized algorithm that computes such polynomials efficiently---the expected running time of our algorithm is linear in $\abs{\S}$. Our approach exploits the technique of \emph{quantifier elimination} combined with that of \emph{$\eps$-samples}. We also present an extension of our construction to \emph{multi-level polynomial partitioning} for semi-algebraic sets in $\RR^d$. 

  We present five applications of our result. The first is a data structure for answering point-enclosure queries among a family of semi-algebraic sets in $\RR^d$ in $O(\log n)$ time, with storage complexity and expected preprocessing time of $O(n^{d+\eps})$. The second is a data structure for answering range-searching queries with semi-algebraic ranges
  in $\RR^d$ in $O(\log n)$ time, with $O(n^{t+\eps})$ storage and expected preprocessing time, where $t > 0$ is an integer that depends on $d$ and the description complexity of the ranges. The third is a data structure for answering vertical ray-shooting queries among semi-algebraic sets in $\RR^{d}$ in $O(\log^2 n)$ time, with $O(n^{d+\eps})$ storage and expected preprocessing time. The fourth is an efficient algorithm for cutting algebraic curves in $\RR^2$ into pseudo-segments. The fifth application is for eliminating depth cycles among triangles in $\RR^3$, where we show a nearly-optimal algorithm to cut $n$ pairwise disjoint non-vertical triangles in~${\RR}^3$ into pieces that form a depth order.
\end{abstract}

\section{Introduction}
\label{introSection}

In 2015, Guth \cite{Guth-15} proved that, if $\S$ is a collection of $n$ $\vardim$-dimensional semi-algebraic sets\footnote{%
  Roughly speaking, a semi-algebraic set in $\RR^d$ is the set of points in $\RR^d$ that satisfy a Boolean formula over a set of polynomial inequalities; the complexity of a semi-algebraic set is the number of polynomials defining the set and their maximum degree.
  See \cite[Chapter 2]{BCR-98} for formal definitions of a semi-algebraic set and its dimension.}
in ${\RR}^d$, each of complexity at most $b$, and if $D\geq 1$ is an integer, then there is a $d$-variate polynomial $P$ of degree at most $D$ so that each connected component of $\RR^d\setminus Z(P)$ intersects $O(n/D^{d-\vardim})$ sets from $\S$,\footnote{%
  Guth stated his result for the special case where the semi-algebraic sets are real algebraic varieties. 
However, since any semi-algebraic set is contained in a variety of the same dimension, Guth's theorem immediately extends to multisets of semi-algebraic sets.}
where the implicit constant in the $O(\cdot)$ notation depends on $d$ and $b$.
We refer to such a polynomial $P$ as a \emph{generalized partitioning polynomial}.
Guth's proof established the existence of a generalized partitioning polynomial, but it did not offer an efficient algorithm to compute such a polynomial for a given collection of semi-algebraic sets.  In this paper we study the problem of computing a generalized partitioning polynomial efficiently and present a few applications of such an algorithm. 

As common in computational geometry, we use the \emph{real} RAM model of computation, where we can compute exactly with arbitrary real numbers and each arithmetic operation is executed in constant time. 
Throughout this paper we assume $d,b,$ and $\eps$ to be constants. Whenever big-O notation is used, the implicit constant may depend on $d,b,$ and $\eps$.

\paragraph{Related work.}
In 2010, Guth and Katz \cite{GK-15} essentially resolved the Erd\H{o}s distinct distances problem in the plane. A major ingredient in their proof was a partitioning theorem for points in $\RR^d$. Specifically, they proved that given a set of $n$ points in $\RR^d$ and an integer $D\geq 1$, there is a $d$-variate \emph{partitioning polynomial} $P$ of degree at most $D$, so that each connected component of~$\RR^d \setminus Z(P)$ contains $O(n/D^d)$ points from the set. Their polynomial partitioning theorem led to a flurry of new results in combinatorial and incidence geometry, harmonic analysis, and theoretical computer science; see e.g.~\cite{Guth-book}.

The Guth-Katz result established the existence of a partitioning polynomial, but it did not provide an efficient algorithm to compute such a polynomial. In \cite{AMS-13}, Agarwal, Matou\v{s}ek, and Sharir developed an efficient randomized algorithm to compute partitioning polynomials, matching the degree bound obtained in \cite{GK-15} up to a constant factor.
They used this algorithm to construct a linear-size data structure that can answer semi-algebraic range queries (see our results subsection for a formal definition of semi-algebraic range query) amid a 
set of $n$ points in $\RR^d$ in time $O(n^{1-1/d}\polylog(n))$, which is near optimal.%

A major open question in geometric range searching is whether a semi-algebraic range query can be answered in $O(\log n)$ time using a data structure of size roughly $n^d$; such a data structure is known only in very special cases, e.g., when query ranges are simplices \cite{Ag-survey}. 
If $t$ is the number of (real-valued) parameters needed to represent query ranges, then
the best known data structure for semi-algebraic range searching uses $O(n^{t+\eps})$ space for $t\le4$ and 
$O(n^{2t-4+\eps})$ space for $t>4$ \cite{Ag-survey}. 

In \cite{AEZ-19}, the last three authors developed an efficient algorithm for constructing a partition of $\RR^3$ adapted to a set of space curves. This partition is not given by a partitioning polynomial, but it shares many of the same properties. For other settings, however, no effective method is known for computing a partitioning polynomial.

\paragraph{Our results.}
Our main result is an efficient algorithm for computing a generalized partitioning polynomial for a family of semi-algebraic sets (Theorem~\ref{efficientlyComputePartition}): Given a set $\S$ of $n$~semi-algebraic sets in~$\RR^d$, 
each of complexity at most $b$ for some constant $b>0$, 
our algorithm computes a generalized partitioning polynomial $P$ of given degree $D$ in randomized expected time $O(n \poly(D)+e^{\poly(D)})$, where the constant of proportionality depends on $b$ and $d$.
We first show an algorithm whose running time is exponential in $n$ using \emph{quantifier elimination} (Theorem~\ref{efficientPartitioningFewVarieties}), and then present a speed-up using 
random sampling
(Lemma~\ref{epsSamplingPartitioningThm}).
Not only does our algorithm compute the polynomial $P$, but it also computes within same time bound a semi-algebraic representation of each connected component of $\RR^d\setminus Z(P)$ and the subset of input sets intersecting each connected component.

We note that neither Guth's result nor our algorithm provide any guarantee on the number of input sets intersecting $V\coloneqq Z(P)$. In fact, all of them might intersect $V$. Although the above polynomial partitioning is sufficient for some applications, others call for a  further partition of $V$ into smaller regions so that each of them is crossed by a relatively small number of sets in $\S$.\footnote{ Given two sets $A$ and $B$, we say that $A$ \emph{crosses}  $B$ if $A\cap B\neq\emptyset$ and $B\not\subset A$.} 
Section~\ref{sec:multi_level} presents an algorithm for computing
such a decomposition (Theorem~\ref{multiLevelPartitionBdry}),
the so-called  \emph{multi-level polynomial partitioning} of $\RR^d$,
by applying Theorem~\ref{efficientlyComputePartition} iteratively, 
Such multi-level partitioning schemes have appeared previously in the literature (see, e.g., \cite{MP, FPSSZ})
but the version developed in this paper has some innovations that generalize
and somewhat simplify previous approaches. A key ingredient of our analysis is
to use \emph{orthogonal projection} and show that 
$V$ can be projected onto a subspace of dimension $\dim(V)$ so that the 
the dimension of the projection of $V$ does not decrease (cf.  Lemma~\ref{goodOrthogonalProjection}). A similar lemma was the key ingredient in the partitioning scheme by Matou\v{s}ek and Pat\'{a}kov\'{a}~\cite{MP}, but our proof is constructive and we present a deterministic algorithm to compute the projection. 
As a consequence of our projection lemma, we can project onto this
subspace all sets $S \cap V$, for any $S \in \S$, construct the
corresponding polynomial partitioning in this subspace, and pull back the polynomial to $\RR^d$.
This result is interesting in its own interest and is likely to be useful for problems beyond those described in this paper.
\smallskip

Finally, Section~\ref{sec:applications} presents five applications of our algorithm:
\begin{description}
	\item[\textbf{\textit{Point-enclosure queries.}}]
		Let $\S$ be a family of $n$ semi-algebraic sets in $\RR^d$, each of constant complexity.
	Each set in $\S$ is assigned a weight that belongs to a semigroup. We 
	present a data structure of size $O(n^{d+\eps})$, for any constant $\eps>0$, that can compute, 
	in $O(\log n)$ time, the cumulative weight of the sets in $\S$ containing a query point.
	The data structure can be constructed in $O(n^{d+\eps})$ randomized expected time. This is a 
	significant improvement over the best known data structure by Koltun \cite{Koltun-04}, 
	for $d > 4$, that uses $O(n^{2d-4+\eps})$ space.

	\item[\textbf{\textit{Semi-algebraic range queries.}}]
	Let $P$ be a set of $n$ points in $\RR^d$, each of which is assigned a weight in a semigroup, and let 
	$\R$ be a (possibly infinite) family of semi-algebraic sets in $\RR^d$, each of constant complexity.  Suppose that 
	  there exists a positive integer $t$ and an injection $f \colon \R\to \RR^t$, so 
	that for each $p \in P$, the set $f(\{R \in \R \mid p\in R\})$ is a semi-algebraic set 
	in $\RR^t$ of complexity at most $b$.
        We can construct in $O(n^{t+\eps})$ randomized expected 
	time a data structure of size $O(n^{t+\eps})$, for any constant $\eps>0$,
	that can compute in $O(\log n)$ time the cumulative weight of $P\cap R$ for a 
	query range $R\in\R$.  

	\item[\textbf{\textit{Vertical ray-shooting queries.}}]
	Given a family $\S$ of $n$ semi-algebraic sets in $\RR^d$,  we present a data structure of size 
	$O(n^{d+\eps})$, for any constant $\eps>0$, that can answer vertical ray-shooting queries in 
		$O(\log^2{n})$ time.  That is, for a query ray $\rho$ in the $(+x_d)$-direction, the data 
		structure returns the first element of $S$ intersecting $\rho$.
		The data structure can be constructed in $O(n^{d+\eps})$ randomized expected time. 

	\item[\textbf{\textit{Cutting planar curves.}}]
		Let $\C$ be a set of $n$ algebraic curves in $\RR^2$, each of degree at most $b$ for some constant $b>0$, such that no two of them share a component. We wish to cut $\C$ into 
		\emph{pseudo-segments}, i.e., into  a set of (open) Jordan arcs so that each pair of arcs intersects at most once.\footnote{Let $\Gamma$ be a set of pairwise-disjoint (open) Jordan arcs, each pair of which have finite intersections. $\Gamma$ is called a \emph{cutting} of $\C$ if each arc in $\Gamma$ lies on a curve of $\C$ and each curve in $\C$ can be expressed as a finite union of pairwise-disjoint arcs from $\Gamma$ plus finitely many points (the points at which the curves are cut).}
		Sharir and Zahl \cite{SZ-17} showed that $\C$ can be cut into  $O(n^{3/2}\log^{O(1)}n)$ pseudo-segments. Following their technique but exploiting
	Theorem~\ref{thm:efficientlyComputePartition}, we show that $\C$ can be cut in 
		$O(n^{3/2 + \eps})$  randomized expected time into $O(n^{3/2+\eps})$ pieces, for any constant $\eps>0$.
        We comment that the technique by Aronov~\etal \cite{AEZ-19} results in the same asymptotic bound
        on the number of pseudo-segments, however, its running time is quadratic in the number of curves. 
		By plugging our algorithm into the technique in~\cite{SZ-17} and using the algorithm by Agarwal~\etal~\cite{AS-05} for computing marked faces in the arrangement of a set of pseudo-segments, we can compute a set of 
		$m$ marked faces in the arrangement of $\C$ in $O(m^{2/3-\eps}n^{2/3+\eps}+m^{1+\eps}+n^{3/2+\eps})$ randomized expected time.

		\item[\textbf{\textit{Eliminating depth cycles.}}]
        Finally, we follow the technique of Aronov~\etal \cite{AMS-19} to eliminate depth cycles among triangles in $3$-space. Specifically, Aronov~\etal \cite{AMS-19} used the generalized partitioning polynomial of Guth \cite{Guth-15} in order to show that $n$ pairwise disjoint non-vertical triangles in~${\RR}^3$ can be cut into~$O(n^{3/2 + \eps})$ pieces that do not form cycles (and thus form a so-called \emph{depth order}), for any $\eps > 0$. Using Theorem~\ref{efficientlyComputePartition} we present an efficient algorithm to produce these pieces in $O(n^{3/2 + \eps})$  randomized expected time.
\end{description}

\section{Preliminaries}
\label{sec:prelim}

In what follows, the \emph{complexity} of 
a semi-algebraic set $S$ in $\RR^d$ is the minimum value $b$ so that $S$ can be represented as the set of points $x \in \RR^d$ satisfying a (quantifier-free) Boolean formula with at most $b$ atoms of the form $P(x) < 0$ or $P(x) = 0$, with each $P$ being a $d$-variate polynomial of degree at most $b$.

Our analysis makes extensive use of concepts and results from real algebraic geometry and random sampling.
We review them below.

\subsection{Polynomials, partitioning, and quantifier elimination}
\label{sec:poly_partitioning}

\paragraph{Sign conditions.}
	Let $\RR[x_1,\ldots,x_d]$ denote the set of all $d$-variate polynomials with real coefficients. 
For $P_1,\ldots,P_\ell \in \RR[x_1,\ldots,x_d]$, a \emph{sign condition} on $(P_1,\ldots,P_\ell)$ is an element of~$\{-1,0,1\}^\ell$. A \emph{strict sign condition} on $(P_1,\ldots,P_\ell)$ is an element of~$\{-1,1\}^\ell$. A sign condition $(\nu_1,\ldots,\nu_\ell)\in\{-1,0,1\}^\ell$ is \emph{realizable} if the set
\begin{equation}
  \label{realizationOfSignCondition}
  \{ x \in \RR^d \mid \sign(P_j(x))=\nu_j \text{ for each }j=1,\ldots,\ell\}
\end{equation}
is non-empty.  A realizable strict sign condition is defined analogously.  The set \eqref{realizationOfSignCondition} is the \emph{realization} of the sign condition. The set of \emph{realizations of sign conditions} (resp., \emph{realizations of strict sign conditions})
corresponding to the tuple $(P_1,\ldots,P_\ell)$ is the collection of all non-empty sets of the above form. Realizations of sign conditions are pairwise disjoint, and the set of realizations of sign conditions partitions $\RR^d$.
We use $Z(\P)=\bigcap_{P\in\P} Z(P)$ to denote the common zero set of all polynomials in $\P$, i.e., $Z(\P)$ is the realization of the sign sequence $(0, \ldots, 0)$.

While a tuple of $\ell$ polynomials has $3^\ell$ sign conditions and $2^\ell$ strict sign conditions, Milnor and Thom showed that the number of realizable sign conditions is often much smaller. We will use the following version of this fact, which is a special case
of  Proposition 7.31 in \cite{BPR} (see also Theorem~2.2 in \cite{FPSSZ}).

\begin{lemma}
  \label{milnorThomResult}
  Let $2\leq d\leq\ell$ and let  $\P=\{P_1,\ldots,P_{\ell}\}$ be polynomials in $\RR[x_1,\ldots,x_d]$ of degree at most~$t$. Then $Z(\P)$ has at most $t(2t-1)^{d-1}$ connected components, and $P_1,\ldots,P_{\ell}$ have at most $(50t\ell/d)^d$ realizable sign conditions.
\end{lemma}

\paragraph{Arrangements.} 
Let $\PP \subset \RR[x_1,\ldots,x_d]$ be a finite set of polynomials. The \emph{arrangement} of the zero sets of $\PP$, denoted by $\A(\PP)$, is the partition of $\RR^d$ into the connected components of the realizations of sign conditions of $\PP$. Equivalently, $\A(\P)$ is the partition of $\RR^d$  into maximal relatively open connected regions, called \emph{cells}, so that for each cell $\tau$, there is a subset $\PP_\tau$ such that $\tau\subseteq Z(P)$ for each $P\in \PP$ and $\tau\cap Z(P) \ne \emptyset$ for all $P\not\in\PP_\tau$.

We recall Theorem~2.16 from~\cite{basu-survey}:
\begin{proposition}
  \label{pointLocationThm}
	Let $\P$ be a set of $s$~polynomials in $\RR[x_1,\ldots,x_d]$ of degree at most $t$. There is an algorithm that computes a set of points $\Q$ in $\RR^d$, so that each cell of $\A(\P)$ contains exactly one point from $\Q$. This algorithm runs in time $s^d t^{O(d)}$. Furthermore, there is an algorithm that computes the sign of each polynomial in $\P$ at each point of $\Q$. This algorithm runs in time $s^{d+1} t^{O(d)}$.
\end{proposition}

The following proposition is a summary of Theorem~16.11 from~\cite{BPR}.

\begin{proposition}
	\label{arrangement}
	Let $\PP\subset\RR[x_1,\ldots,x_d]$ be a set of $s$ polynomials, each of degree at most $t$. Then the arrangement of the zero sets of $\PP$ has $s^dt^{O(d)}$ cells, and it can be computed in time $T=s^{d+1}t^{O(d^4)}$. Each cell is described as a semi-algebraic set using at most $T$ polynomials of degree at most $t^{O(d^3)}$. Moreover, the algorithm computes a reference point in each cell as well as the adjacency information for the cells, indicating which cells are contained in the boundary of each cell.
\end{proposition}

We also need the following result from~\cite{BB-12}:
\begin{proposition}
	\label{BB-bound}
	Let $\P$ be a set of $s$~polynomials in $\RR[x_1,\ldots,x_d]$, each of  of degree at most $t$, let 
	$V\coloneqq Z(\P)$, and let $\dim V\coloneqq k$. let $\G$ be a set of $m$ polynomials of degree at most $E\ge t$. Then the number of cells of $\A(\P\cup\G)$ contained in  $V$ is $O(1)^d t^{d-k}(mE)^k$.
\end{proposition}

\paragraph{Polynomials and partitioning.}
The partitioning polynomial Guth constructed in \cite{Guth-15} for a family~$\S$ of $n$~semi-algebraic sets in $\RR^d$ is the product of a set of polynomials $P_1\cdots P_{\ell}$, where each $P_k$, $1 \leq k \leq \ell$ has a prescribed degree, and the realization of each of the $2^{\ell}$ strict sign conditions of $\{P_1,\dots,P_{\ell}\}$ intersects few semi-algebraic sets from $\S$. Guth used a variant of the Borsuk-Ulam theorem to guarantee that such a tuple of polynomials exists.

For each positive integer $j$, let $D_j$ be the smallest positive integer so that $\binom{D_j+d}{d} > 2^{j-1}$; we have $D_j\leq d2^{j/d}$. In what follows, we are interested in tuples of $d$-variate polynomials of the form $(P_1,\ldots,P_k)$, where for each index $j$, $P_j$ has degree at most $D_j$. 

Let $(P_1,\ldots,P_k)$ be such a tuple, let $\S$ be a collection of $n$  semi-algebraic sets in $\RR^d$, and let $\alpha \geq 1$. We say that $(P_1,\ldots,P_k)$ is a \emph{$(k,\alpha)$-partitioning tuple} for $\S$ if the realization of each sign condition in $(P_1, \ldots, P_k)$ intersects at most $|\S|/\alpha$ sets from $\S$.%
\footnote{As in \cite{Guth-15}, we work with a $k$-tuple of polynomials instead of a single polynomial so that we can bound the number of sets intersected by the realization of a sign condition rather than by a connected component of a realization.} Guth \cite{Guth-15} proved that if $\alpha$ is chosen appropriately then a $(k,\alpha)$-partitioning tuple is guaranteed to exist:
\begin{proposition}[Generalized Polynomial Partitioning \cite{Guth-15}]
  \label{guthProp}
  Let $\S$ be a family of semi-algebraic sets in $\RR^d$, each of dimension at most $\vardim$ and complexity at most $b$.
  For each $k\geq 1$, there exists a $(k,\alpha)$-partitioning tuple for $\S$, with $\alpha = \Omega(2^{k(1-\vardim/d)})$. 
\end{proposition}

\noindent\textbf{Remark.}
(i) In \cite{Guth-15}, Guth proved a variant of Proposition \ref{guthProp} where $\S$ is a (multi) set of algebraic varieties in $\RR^d$, each of dimension at most $\vardim$. However, since any semi-algebraic set is contained in a variety of the same dimension, Guth's theorem immediately extends to multisets of semi-algebraic sets.

(ii) The big-Omega notation in the bound on $\alpha$ hides a factor $1/b^{O(d)}$.

\paragraph{Singly exponential quantifier elimination.}
Let $h$ and $\ell$ be non-negative integers and let $\P = \{P_1,\ldots,P_s\}\subset\RR[x_1,\ldots,x_h,y_1,\ldots,y_\ell]$.
Let $\Phi(y)$ be a first-order formula given by
\begin{equation}
  \label{typeOfPhi}
	\Phi(y) = (\exists x_1, \ldots, x_h) F(P_1(x,y),\ldots,P_s(x,y)),
\end{equation}
where
$y =(y_1,\ldots,y_\ell)$ is a block of $\ell$ free variables; $x$ is a block of $h$ variables, and $F(P_1,\ldots,P_s)$ is a quantifier-free Boolean formula with atomic predicates of the form $\sign(P_i(x_1,\ldots,x_h, y)) = \sigma$, with $\sigma \in \{-1,0,1\}.$

The Tarski-Seidenberg theorem states that the set of points $y \in \RR^{\ell}$ satisfying the formula~$\Phi(y)$ is semi-algebraic. Proposition \ref{BasuBlockElimination} below is a quantitative version of this result that bounds the number and degree of the polynomial equalities and inequalities needed to describe the set of points satisfying $\Phi(y)$. This proposition is known as a ``singly exponential quantifier elimination,'' and its more general form (where $\Phi(y)$ may contain a mix of $\forall$ and $\exists$ quantifiers) can be found in~\cite[Theorem~2.27]{basu-survey}.
\begin{proposition}
  \label{BasuBlockElimination}
	Let $\P \subset \RR[x_1,\ldots,x_h,y_1,\ldots,y_\ell]$ be a set of at most $s$ polynomials, each of degree at most $t$.  Given a formula $\Phi(y)$ of the form \eqref{typeOfPhi}, there exists an equivalent quantifier-free formula
  \begin{equation}\label{quantifierFree}
    \Psi(y) = \bigvee_{i=1}^I \bigwedge_{j=1}^{J_i}\Big(\bigvee_{n=1}^{N_{i,j}}\sign(P_{ijn}(y)) = \sigma_{ijn}\Big),
  \end{equation}
  where $P_{ijn}$ are polynomials in the variables $y$, $\sigma_{ijn}\in\{-1,0,1\}$, 
  \begin{equation}
    \begin{aligned}
      I&\leq s^{(h + 1) (\ell+1)} t^{O(h \cdot \ell)}, \\
      J_i&\leq s^{(h + 1)} t^{O(h)},\\
      N_{ij}&\leq t^{O(h)},
    \end{aligned}
  \end{equation}
  and the degrees of the polynomials $P_{ijn}(y)$ are bounded by $t^{O(h)}$.
  Moreover, there is an algorithm to compute $\Psi(y)$ in time $s^{(h + 1) (\ell+1)} t^{O(h \cdot \ell)}$.
\end{proposition}

\subsection{Range spaces, VC dimension, and sampling}
\label{sec:eps_approx}

We first recall several standard definitions and results from \cite[Chapter 5]{Har-Peled-11}. 
A \emph{range space} (also known as a set system, or a hypergraph) is a pair $\Sigma=(X,\R)$, where $X$ is a set and $\R$ is a collection of subsets of $X$.
Let $(X,\R)$ be a range space and let $A\subset X$ be a set. We define the \emph{restriction} of $\Sigma$ to $A$, denoted by $\Sigma_A$ to be $(A,\R_A)$, where
$\R_A \coloneqq \{R \cap A  \mid R \in\R\}.$
If $A$ is finite, then $\abs{\R_A} \leq 2^{|A|}$. If equality holds, then we say $A$ is \emph{shattered}.
We define the \emph{shatter function} by
$\pi_{\R}(z) \coloneqq \max_{|A| = z} |\R_A|$.
The \emph{VC dimension} of $\Sigma$,  denoted by $\operatorname{VC-dim}(\Sigma)$, is the largest cardinality of a set shattered by $\R$.
If arbitrarily large finite subsets can be shattered, we say that the VC dimension of $\Sigma$ is infinite.

Let $\Sigma$ be a range space, $A$ a finite subset of $X$, and $0\leq \eps\leq 1$.
A set $B\subset A$ is an \emph{$\eps$-sample} (also known as \emph{$\eps$-approximation}) of $\Sigma_A$ if 
\[
  \Big|\frac{|A\cap R|}{|A|}-\frac{|B\cap R|}{|B|}\Big|\leq\eps \quad \forall R\in\R.
\]
The following classical theorem of Vapnik and Chervonenkis \cite{VC} guarantees that, if the VC-dimension of $\Sigma$ is finite, then for each positive $\eps > 0$, a sufficiently large random sample of $A$ is likely to be an $\eps$-sample.\footnote{The stated bound is not the strongest possible (see, e.g., \cite[Chapter 7]{Har-Peled-11} for an improved bound), but is sufficient for our purposes.}
  
\begin{proposition}[$\eps$-Sample Theorem]
  \label{epsSamplingTheorem}
  Let $\Sigma=(X,\R)$ be a range space of VC dimension at most $d$ and let $A\subset X$ be finite. Let $0<\eps,\delta<1$. Then a random subset $B\subset A$ of cardinality $\frac{8d}{\eps^2}\log\frac{1}{\eps\delta}$ is an $\eps$-sample for $\Sigma_A$ with probability at least $1-\delta$. 
\end{proposition}

We will also need the following simple result, which is a converse of the celebrated Sauer-Shelah lemma.
\begin{lemma}
  \label{SauerConverse}
  Let $\rho>0,C>0$, and let $\Sigma=(X,\R)$ be a range space whose shatter function $\pi_{\R}(z)$ satisfies the bound $\pi_{\R}(z) \leq Cz^{\rho}$ for all positive integers $z$. Then  $\Sigma$  has VC dimension at most $4{\rho} \log(C \rho)$. 
\end{lemma}
\begin{proof}
Suppose that $\Sigma$ has VC dimension $d$. Then $2^d=  \pi_{\R}(d) \leq Cd^{\rho}$. Rearranging, we conclude that $d\leq 4{\rho} \log(C \rho)$.
\end{proof}

We next closely follow the arguments in the proof of Corollary 2.3 from \cite{FPSSZ} to show the following.

\begin{proposition}
  \label{vcDimSemiAlgRange}
Let $V\subset\RR^s\times\RR^t$ be a semi-algebraic set of complexity $b$.
  For each $y\in\RR^t$, define $R_y \coloneqq \{x \in \RR^s \mid (x,y) \in V\}$. Then the range space $(\RR^s,\{R_y \mid y \in\RR^t\})$ has VC dimension at most $200t^2\log b$. 
\end{proposition}

\begin{proof}
  By assumption, there are polynomials $f_1,\ldots,f_b$ and a Boolean formula~$\Phi$, so that, for $(x,y)\in \RR^s\times\RR^t$, $(x,y)\in V$ if and only if $\Phi(f_1(x,y)\geq 0,\ldots, f_b(x,y)\geq 0)=1$. 

Put $\R \coloneqq \{R_y \mid y \in \RR^t\}$.
Fix a positive integer $z$ and let $p_1,\ldots,p_z\in \RR^s$. Our goal is to bound
\[
\left|\strut\{R \cap \{p_1,\ldots,p_z\} : R \in \R\}\right| = \left|\strut\{R_y \cap \{p_1,\ldots,p_z\} : y\in\RR^t\}\right|.
\]
For each $j=1,\ldots,z,$ define
\[
W_j \coloneqq \{y \in \RR^t \mid \Phi(f_1(p_j,y) \geq 0, f_2(p_j,y) \geq 0, \ldots, f_b(p_j,y) \geq 0) = 1\}.
\]  

Let $A\subset[z]\coloneqq\{1,\ldots,z\}$ and suppose that there exists $y\in\RR^t$ with $y\in W_j$ for each $j\in A$ and $y\not\in W_j$ for each $j\in [z]\setminus A$, i.e., the semi-algebraic set $S_A$ consisting of those points $y\in\RR^t$ satisfying the Boolean formula
\begin{equation*}
  \begin{split}
    &\bigwedge_{j\in A}\big(\Phi(f_1(p_j,y)\geq 0,f_2(p_j,y)\geq 0,\ldots,f_b(p_j,y)\geq 0)=1\big)\ \\
    &\quad \wedge\ \bigwedge_{j\in[z]\setminus A}\big(\Phi(f_1(p_j,y)\geq 0,f_2(p_j,y)\geq 0,\ldots,f_b(p_j,y)\geq 0)=0\big)
  \end{split}
\end{equation*}
is non-empty. Observe that if $A$ and $A^\prime$ are distinct subsets of $[z]$, then $S_A$ and $S_{A^\prime}$ are disjoint and, in fact,
\[
  \left|\strut\{R_y \cap \{p_1,\ldots,p_z\} : y\in\RR^t\}\right| =
  \left|\strut \{A \subset [z] : S_A \neq \emptyset\}\right|. 
\]
Each of the non-empty sets $S_A$ contains at least one realization of a sign condition of the
$bz$ polynomials 
\[
\{f_i(p_j,y) \mid i=1,\ldots,b;\ j=1,\ldots,z\},
\]
each of degree at most $b$.
By Lemma \ref{milnorThomResult}, these polynomials determine at most $(50\cdot b\cdot bz/t)^t\leq (50b^2z)^t$ realizable sign conditions. Thus 
\begin{equation}
  \label{boundOnShattering}
  |\{R_y \cap \{p_1,\ldots,p_z\} : y \in \RR^t\}|\leq (50b^2z)^t.
\end{equation}
Since \eqref{boundOnShattering} holds for every choice of $p_1,\ldots,p_z\in\RR^s$, we conclude that
\[
\pi_{\R}(z)\leq (50b^2z)^t.
\]
By Lemma \ref{SauerConverse}, $(\RR^s, \{R_y \mid y \in \RR^t\})$ has VC dimension at most $4t\log((50b^2)^t t)\leq 200t^2\log b$.
\end{proof}

\section{Computing a Generalized Polynomial Partition} %
\label{proofOfThmEpsSamplingPartitioningThmSec}

In this section we obtain the main result of the paper: given a collection $\S$ of semi-algebraic sets in~$\RR^d$, each of dimension at most~$\vardim$ and complexity at most $b$, a~$(k,\Omega(D^{d-\vardim}))$-partitioning tuple for~$\S$ can be computed efficiently. We obtain this result in several steps.  Given a semi-algebraic set~$S$, a sign condition $\sigma\in \{-1,+1\}^k$, and a positive integer $b>0$, we first show that the set of $k$-tuples of degree-$b$ polynomials whose realization of $\sigma$ intersects $S$ is a semi-algebraic set.  This in turn implies that, if $S_1,\ldots,S_n$ are semi-algebraic sets and if $m\leq n$ is a parameter, then the set of $k$-tuples of degree-$b$ polynomials 
whose realization intersects at most $m$ of the sets $S_1,\ldots,S_n$ is semi-algebraic. We use a quantifier-elimination algorithm to find a desired $k$-tuple. 
Unfortunately, the running time of the algorithm is exponential in $n$. We reduce the running time of the algorithm by using a random sampling technique --- we show that it suffices to compute a partitioning tuple with respect to a small-size random subset of~$\S$.

\subsection{The parameter space of polynomials}
\label{sec:poly-para}
In this section we will discuss the parameter space of tuples of polynomials $(P_1,\ldots,P_k)$ that are candidates for the $(k,\alpha)$-partitioning tuple from Proposition \ref{guthProp}.

The set of polynomials in $\RR[x_1,\ldots,x_d]$ of degree at most~$b$ is a real vector space of dimension~$\binom{b+d}{d}$; we identify this vector space with $\RR^{\binom{b+d}{d}}$. For a point~$q\in \RR^{\binom{b+d}{d}}$, let $P_q\in\RR[x_1,\ldots,x_d]$ be the corresponding polynomial of degree at most $b$.  An important relationship between the vector space $\RR^d$ and the set of $d$-variate polynomials of degree at most $b$ can be expressed using a polynomial that we call  $Q(q,x)$. For $q\in \RR^{\binom{b+d}{d}}$ and $x\in\RR^d$, define $Q(q,x) \coloneqq P_q(x)$. 
Since we can write $Q(q,x) = \sum_{i=1}^{\binom{b+d}{d}}q_i H_i(x)$, where $H_i$ is a monomial of degree at most $b$, the following observation, which will be useful later, is straightforward:

\begin{lemma}
  \label{specialPoly}
The polynomial $Q$ has degree $b+1$.
\end{lemma}

For a positive integer $j$, we recall that $D_j$ is the smallest positive integer so that 
$$\binom{D_j+d}{d} > 2^{j-1};$$ 
we have $D_j\leq d2^{j/d}$;
we will often write this as $D_j = O(2^{j/d})$.
Let $k$ be a positive integer.  Define the product space
\begin{equation}
  \label{defnOfY}
\YY_k \coloneqq \bigtimes_{j=1}^k \RR^{\binom{D_j+d}{d}}.
\end{equation}
We identify each point $y=(y_1, \ldots, y_k) \in \YY_k$, where $y_j \in \RR^{\binom{D_j+d}{d}}$,
with a $k$-tuple of $d$-variate polynomials $\mathbf{P}_y=(P_{y_1},\ldots,P_{y_k})$. For each $j=1,\ldots,k$, 
$$\deg(P_j) \leq D_j =  O(2^{j/d})$$
and thus $\deg\big(\prod_{j=1}^k P_j\big)=O(2^{k/d})$.

\subsection{The parameter space of semi-algebraic sets}
\label{sec:semi-para}

Fix positive integers $b$, $d$, $\vardim$, and $k$, and let $D \coloneqq 2^{k/d}$.
Hereafter we assume that $D = \Omega(2^b)$, which can be enforced by choosing $k$ sufficiently large.

As above, let $\S$ be a family of semi-algebraic sets in $\RR^d$, each of dimension at most~$\vardim$ and complexity at most $b$. 
Let $G\colon\{0,1\}^b\to\{0,1\}$ be a Boolean function.
Let $\XX \coloneqq \big(\RR^{\binom{b+d}{d}}\big)^b$. We identify a point $x = (q_1,\ldots,q_b)\in \XX$ with the semi-algebraic set 
\[
  Z_{x,G} \coloneqq \{v\in \RR^d \mid G(P_{q_1}(v)\geq 0,\ldots,P_{q_b}(v)\geq 0)=1\} \subset\RR^d.
\]
Observe that each semi-algebraic set in $\S$ is of the form $Z_{x,G}$ for some choice of $x\in \XX$ and a Boolean function $G$.
Let $\YY \coloneqq \YY_k$ (see~\eqref{defnOfY}).
For each $y\in \YY$, define $S_y\coloneqq\{u \in \RR^d \mid P_1(u) > 0, \ldots, P_k(u) > 0\}$,
where $(P_1,\ldots,P_k)$ is the tuple associated with $y$.
Define
\[
  W_G \coloneqq \{(x,y)\in \XX \times \YY \mid Z_{x,G} \cap S_y \neq \emptyset\}.
\]

\begin{proposition}
  \label{WHasLowComplexity}
  The set $W_G$ is semi-algebraic; it is defined by $O(e^{\poly(D)})$ polynomials, each of degree~$D^{O(d)}$.
\end{proposition}

\begin{proof}
Define $\ZZ \coloneqq \{(x,y,v)\in \XX \times \YY \times \RR^d \mid v \in Z_{x,G}\cap S_y \}$.
The condition $v\in Z_{x,G}$ is a Boolean condition on $b$ polynomial inequalities in $\RR^d$. By Remark~\ref{specialPoly}, each of these polynomials has degree at most~$b+1$, where the coefficients $x = (q_1,\ldots,q_b)\in \XX$ are now viewed as variables in addition to $v$. Similarly, the condition $v\in S_y$ consists of $k$ polynomial inequalities in $\RR^d$, each of degree at most~$D+1$, where the coefficients $y = (y_1,\ldots,y_k)\in \YY$ are now viewed as variables in addition to $v$ .

This means that there exists a set of polynomials $\Q = \{Q_1,\ldots,Q_{b+k}\}$ 
of degree~$b+D+1$ in the variables $x,y,v$, and a Boolean function $F(z_1,\ldots,z_{b+k})$ so that
\[
\ZZ = \{(x,y,v)\in \XX \times \YY \times \RR^d \mid F(Q_1(x,y,v)\ge 0, \ldots, Q_{b+k}(z,y,v)\ge 0) = 1\}.
\]
With the above definitions 
\[
W_G = \{(x,y) \mid \exists v\, F(Q_1(x,y,v),\ldots,Q_{b+k}(x,y,v)) = 1\}.
\]

We now apply Proposition~\ref{BasuBlockElimination}. We have a set $\Q$ of $s \coloneqq b+k$ polynomials, each of degree at most~$t \coloneqq b+D+1$.
The variables $h$ and $\ell$ from the hypothesis of Proposition~\ref{BasuBlockElimination} are set to $h \coloneqq d$ and $\ell \coloneqq O\left( {\binom{b+d}{d}}^b + D^{d}\right) = \poly(D)$; 
recall that $D$ is sufficiently larger than $b$, and thus $\ell$ is a suitably chosen polynomial function of $D$.
With these assignments,  Proposition~\ref{BasuBlockElimination} says that $W_G$ can be expressed as a quantifier-free formula of the form
\begin{equation}
  \label{WQuantFree}
  \bigvee_{i=1}^I \bigwedge_{j=1}^{J_i}\Big(\bigvee_{n=1}^{N_{i,j}}\sign(P_{ijn}(x,y)) = \sigma_{ijn}\Big),
\end{equation}
where $P_{ijn}$ are polynomials in the variables $(x,y)$, $\sigma_{ijn}\in\{-1,0,1\}$,
\begin{equation}
\begin{aligned}
I&\leq (b+k)^{\poly(D) }(b+D)^{\poly(D)} = O(e^{\poly(D)}),\\
J_i&\leq (b+k)^{d+1} (b+D)^{O(d)} = O(D^{O(d)}),\\
N_{ij}&\leq (b+D)^{O(d)} = O(D^{O(d)}),
\end{aligned}
\end{equation}
where the degrees of the polynomials $P_{ijn}(y)$ are bounded by $(b+D)^{O(d)} = D^{O(d)}$.

Summarizing, the quantifier-free formula \eqref{WQuantFree} for $W_G$ is a 
Boolean combination of $O(e^{\poly(D)})$ polynomial inequalities, each of degree~$D^{O(d)}$, as claimed.
\end{proof}

\subsection{A singly-exponential algorithm}
\label{sec:singly_exp_alg}

In this section, we discuss how to compute a $(k,\alpha)$-partitioning tuple (for an appropriate value of~$\alpha$) for a small number $m$ of semi-algebraic sets.

\begin{theorem}
  \label{efficientPartitioningFewVarieties}
  Let $\S$ be a family of $n$ semi-algebraic sets in $\RR^d$, each of
  dimension at most $\vardim$ and complexity at most $b$. Let $1\leq k\leq\log n$ and let $D \coloneqq 2^{k/d}$. Then a $(k,\Omega(D^{d-\vardim}))$-partitioning tuple for $\S$ can be computed in $O(e^{\poly(m)})$ time.
\end{theorem}

\begin{proof}
  Set $\YY \coloneqq \YY_k$. As above, we identify points in $\YY$ with $k$-tuples $(P_1,\ldots,P_k)$ of polynomials. Since the class of semi-algebraic sets is closed under projections, the argument in Proposition \ref{WHasLowComplexity} shows that, for each $S \in \S$ and each $\sigma \in \{-1,1\}^k$, 
\[
I_{S,\sigma} \coloneqq \{y \in \YY \mid S \cap\{\sigma_1 P_1 > 0, \sigma_2 P_2 > 0, \ldots, \sigma_k P_k > 0\} \neq\emptyset \}
\]
is a semi-algebraic set in $\YY$ that can be expressed as a Boolean combination of $O(e^{\poly(D)})$ polynomials, each of degree $D^{O(d)}$. Moreover, it can be computed in time $O(e^{\poly(D)})$ (see
Proposition~\ref{BasuBlockElimination}).

Let $C_{b,d}$ be a constant to be specified later (the constant will depend only on $b$ and $d$) and let $N \coloneqq C_{b,d}nD^{\vardim-d}+1$; observe that $N = O(n)$. For each $\sigma\in\{-1,1\}^k$ and for each set $\S^\prime \subset \S$ of cardinality $|\S^\prime| \ge N$, the set
$\{y \in \YY \mid y \in I_{S, \sigma}\text{ for every } S \in\S^\prime\}$
is a semi-algebraic set in $\YY$ that can be expressed as a Boolean combination of $O(N^\prime e^{\poly(D)}) = O(n e^{\poly(D)})$ polynomials, each of degree $D^{O(d)}$, where $N' \coloneqq \abs{\S^\prime}$.
Therefore
\begin{equation}
  \label{badSet}
  \K \coloneqq
	\bigcup_{\substack{{S}^\prime \subset \S \\ |\S^\prime| \ge N}} \{y \in \YY \mid y\in I_{S,\sigma}\text{ for every } S \in \S^\prime\}
\end{equation}
is a semi-algebraic set in $\YY$ that can be expressed as a Boolean combination of
\[
  \sum_{\substack{{S}^\prime \subset \S \\ |\S^\prime| \ge N}} O\Biggl(\binom{n}{|S^\prime|} n e^{\poly(D)}\Biggr) = O(e^{n + \poly(D)}) = O(e^{\poly(n)})
\]
polynomials, each of degree $D^{O(d)}$. Since the complement of a semi-algebraic set is also semi-algebraic, we conclude that 
\[
\operatorname{Good}(\sigma) \coloneqq \YY \setminus \K =
\{y \in \YY \mid y\in I_{S,\sigma}\text{ for fewer than } C_{b,d}|\S|D^{\vardim-d} \text{ sets } S \in \S\}
\]
is a semi-algebraic set in $\YY$ that can be expressed as a Boolean combination of $O(e^{\poly(n)})$ 
polynomials,
each of degree $D^{O(d)}$.
This means that the set
\begin{equation}
  \label{goodPartitioningTuplesSet}
  \bigcap_{\sigma\in\{-1,1\}^k}\operatorname{Good}(\sigma)
\end{equation}
is a semi-algebraic set in $\YY$ that can be expressed as a Boolean combination of $O(e^{\poly(n)})$ polynomials, each of degree $D^{O(d)}$.
Recall that, by assumption, $1 \leq k \leq \log n$ and $D = 2^{k/d}$. It thus follows that the degree is bounded by $\poly(n)$.
Similarly, the dimension of the space $\YY$ is bounded by $\poly(n)$ as well.

Proposition~\ref{guthProp} guarantees that if $C_{b,d}$ is selected sufficiently large, then the set \eqref{goodPartitioningTuplesSet} is non-empty.
By Proposition~\ref{pointLocationThm}, it is possible to locate a point in this set in $O(e^{\poly(n)})$ time,
concluding the proof of the theorem.
\end{proof}

\subsection{Speeding up the algorithm using $\eps$-sampling}
\label{sec:speeding_up}

In this section we first state and prove a result that allows us to reduce a large collection of semi-algebraic sets to a small collection using random sampling. Results of this type are standard in the range-searching literature, but we provide a proof tailored to the current setup.

\begin{lemma}
  \label{epsSamplingPartitioningThm}
  For every choice of positive integers $b$ and $d$, there is a constant $C\coloneqq C_{b,d}$ so that the following holds: Let $C_0$ be a positive integer. Let $\S$ be a collection of $n$ semi-algebraic sets in $\RR^d$, each of dimension at most $0 \le \vardim < d$ and complexity at most $b$. Let $k$ be a positive integer and let $D \coloneqq 2^{k/d}$ (as above, assume $D = \Omega(2^b)$). Let $B\subset \S$ be a randomly chosen subset of $\S$ of size at least $C D^C$ and let $(P_1,\ldots,P_k)$ be a $(k,\frac{D^{d-\vardim}}{C_0})$-partitioning tuple for $B$. Then with probability at least $1/2$,
  the realization of each of the $D^d$ %
  sign conditions of $(P_1,\ldots,P_k)$ intersects $O(C_0 n D^{\vardim-d})$ elements from $\S$, i.e., $(P_1,\ldots,P_k)$ is a $(k,\Omega(\frac{D^{d-\vardim}}{C_0}))$-partitioning tuple for $B$.
\end{lemma}

\begin{proof} %
Define $\XX$ and $\YY$ as above, and let $G \colon \{0,1\}^b \to \{0,1\}$. For each $y\in \YY$, define the range
\[
R_{y,G} = \{x \in \XX \mid Z_{x,G} \cap S_y \neq \emptyset\}.
\]
Define $\R_G = \{R_{y,G} \mid y \in \YY\}$.
By Proposition \ref{vcDimSemiAlgRange},
the range space $(\XX,\R_G)$ has VC dimension $O(\poly(D))$. Indeed, this follows by writing $\YY = \RR^t$, where $t = O(D^d)$, and using the upper bound $200t^2\log b$ asserted by Proposition \ref{vcDimSemiAlgRange} as well as the assumption $D = \Omega(2^b)$.
Define $\R \coloneqq \bigcup_G \R_G$, where the union is taken over the $2^{2^b}$ Boolean functions $G\colon\{0,1\}^b\to\{0,1\}$. Since $R$ is the union of $2^{2^b}$ set systems, the shatter function grows
by at most a multiplicative factor of $2^{2^b}=O(1)$, and therefore the VC dimension of the range space $\Sigma = (\XX,\R)$ is also $O(\poly(D))$ (this is a standard fact, see, e.g., \cite[Chapter~5]{Har-Peled-11}).

We are now ready to prove the statement of the lemma.
Set $\eps \coloneqq C_1 \cdot C_0 D^{\vardim-d}$, where $C_1 > 1$ is an absolute constant.  
Suppose that $B$ is an $\eps$-sample of $\S$ and that $(P_1,\ldots,P_k)$ is a $(k,\frac{D^{d-\vardim}}{C_0})$-partitioning tuple for $B$.
Then for each range $R \in \R$, we have $|B\cap R| \le |B| C_0 D^{\vardim - d}$. 
Recalling the $\eps$-sample property
\[
\Big|\frac{|\S\cap R|}{|\S|}-\frac{|B\cap R|}{|B|}\Big| \leq \eps 
\]
and the choice of $\eps$ (recall that $C_1 > 1$), we obtain
\[
\Big|\frac{|\S\cap R|}{|\S|}\Big|\leq 2\eps,
\]
and thus
\[
|\S\cap R| = O( C_0 n D^{\vardim - d}).
\]
By Proposition~\ref{epsSamplingTheorem}, it is sufficient to pick $B$ of size
\[
\frac{\operatorname{VC-dim}(\Sigma)}{\eps^2}\log(2/\eps) = O(\poly(D)). \qedhere
\]
\end{proof}

Intuitively, Lemma~\ref{epsSamplingPartitioningThm} states that it is sufficient to consider a random subset $B$ of size polynomial in $D$ in order to obtain an appropriate partitioning tuple for the entire collection~$\S$, with reasonable probability.

We next proceed as follows.
We select a random sample of $\S$ of cardinality $CD^C$ and use Theorem~\ref{efficientPartitioningFewVarieties} to compute the corresponding partitioning tuple $\P = (P_1,\ldots,P_k)$ in $O(e^{\poly(D)})$ time.
By Lemma~\ref{epsSamplingPartitioningThm}, this tuple will be a $(k,\Omega(D^{d-\vardim}))$-partitioning tuple for $\S$ with probability at least $1/2$. 
We verify in $O(n\poly(D))$ time whether $\P$ is a partitioning tuple for $\S$, as described below. If $\P$ does not produce the appropriate partition, we discard it and try again; the expected number of trials is at most $2$.

The verification step is done as follows. For each semi-algebraic set $S \in \S$ we compute the subset of sign conditions of $(P_1,\ldots,P_k)$, with which it has a non-empty intersection. To this end, we restrict each of the polynomials $P_1,\ldots,P_k$ to $S$ and apply Proposition~\ref{pointLocationThm} to this restricted collection, 
thereby obtaining a set of points meeting each
connected component of each of the 
realizable sign conditions, as well as the corresponding list of signs of the restricted polynomials for each of these points.
This is done in 
$D^{O(d)}$ time for a single semi-algebraic set $S \in \S$, and overall 
$n D^{O(d)}$ time, over all sets.
We refer the reader to \cite{BB-12} for further details concerning the complexity of the restriction of $P_1,\ldots,P_k$ to $S$.
We have thus shown the following:

\begin{theorem}
  \label{efficientlyComputePartition}
  Let $\S$ be a collection of $n$ semi-algebraic sets in $\RR^d$, each of which has dimension at most $\vardim$ and complexity at most $b$. Let $k\geq 1$ and let $D \coloneqq 2^{k/d}$. Then a $(k,\Omega(D^{d-\vardim}))$-partitioning tuple for $\S$ can be computed in $O(n\poly(D) + e^{\poly(D)})$ ranomized expected time.
\end{theorem}

Let $\P = (P_1, \ldots, P_k)$ be the partitioning tuple obtained by Theorem~\ref{efficientlyComputePartition}, and let $P=\prod_{i=1}^k P_i$.  Let $\Omega$ be the set of connected components, also called \emph{cells}, of $\RR^d\setminus Z(P)$.
By 
Warren \cite[Theorem~2]{Warren-68}, 
the number of cells in $\Omega$ is $O(D^d)$, and a
semi-algebraic representation of each of these cells can be computed $D^{O(d^4)}$ time (cf. Proposition~\ref{arrangement}). Finally, for each semi-algebraic set $S\in\S$, we compute in $O(\poly(D))$ time the cells of $\Omega$ that it crosses. Hence, we compute in $O(n\poly(D))$ time the subset of $\S$ crossed by every cell of $\Omega$. 
We thus conclude:

\begin{theorem}
  \label{thm:efficientlyComputePartition}
  Let $\S$ be a collection of $n$ semi-algebraic sets in $\RR^d$, each of which has dimension at most $\vardim$ and complexity at most $b$. 
  Then for any $D \ge 1$, there is a non-zero polynomial $P$ of degree at most $D$, so that each cell of $\RR^d \setminus Z(P)$ crosses $O(n/D^{d-\vardim})$ elements of $\S$.
  The polynomial $P$ and the semi-algebraic representation of every cell of $\RR^d \setminus Z(P)$ along with the sets of $\S$ crossing it can be computed in $O(n\poly(D) + e^{\poly(D)})$ randomized expected time.
\end{theorem}

\section{Multilevel partitioning}
\label{sec:multi_level}

Recall that the polynomial produced by Theorem~\ref{thm:efficientlyComputePartition} partitions $\RR^d$ into a collection of open semi-algebraic cells plus a ``boundary'' $Z(P)$, so that each cell intersects few semi-algebraic sets from the pre-specified collection $\S$ but $V \coloneqq Z(P)$ may cross  many input sets. 
The main result of this section is an algorithm for computing a  multi-level partitioning 
that decomposes $V$ into smaller regions so that each region is crossed by few of the input 
semi-algebraic sets in $\S$ (Theorem~\ref{multiLevelPartitionBdry}). We also extend the algorithm to the setting in which we are given a set $\S$ of semi-algebraic sets and a set $\Q$ of points in $\RR^d$ and would like to obtain a partition of $\RR^d$ so that each cell contains few points of $\Q$ and crosses few sets of $\S$ (Corollary~\ref{cor:multiLevelPartitionBdry}).

\paragraph{Relative boundary and orthogonal projection.}
If $X$ and $Y$ are subsets of $\RR^d$ with $X\subset Y$, we define the \emph{relative boundary} of $X$ in $Y$, denoted by $\bd_Y X$, to be the set
\begin{equation}
  \label{relBoundary}
	\bd_Y X \coloneqq \{y\in Y \mid \forall\eps>0\, \exists x_1,x_2\in \RR^d : \|y-x_1\|,\|y-x_2\|\le \eps \wedge (x_1\in X) \wedge (x_2 \in Y\backslash X)\}.
\end{equation}

In general, $\bd_Y X$ might have the same dimension as $Y$. The next lemma shows that when $X$ and $Y$ are semi-algebraic sets, this cannot be the case.

\begin{lemma}
  \label{complexityOfRelativeBdry}
	Let $X\subset Y\subset\RR^d$ be semi-algebraic sets of complexity $b$ and $E$ respectively. Then $\bd_Y X$ is a semi-algebraic set of dimension at most $\min\{\dim(X),\dim(Y)-1\}$ and complexity $(bE)^{O(d)}$, and it can be computed in time $(bE)^{O(d)}$ time.
\end{lemma}
\begin{proof}
	The first part of the lemma is essentially Proposition 2.8.13 in Bochnak~\etal~\cite{BCR-98}. By applying  the singly exponential quantifier-elimination algorithm to the definition of $\bd_Y X$ in \eqref{relBoundary} (cf.\ Theorem~14.16 in \cite{BPR}), we construct in $(bE)^{O(d)}$ time a  quantifier-free formula of complexity $(bE)^{O(d)}$. 
\end{proof}

Next, we state a standard result in real algebraic geometry, which we will be using repeatedly:

\begin{lemma}
	\label{lemm:project}
	Let $S$ be a semi-algebraic set in $\RR^d$ of complexity $b$, and let $\pi: \RR^d \rightarrow \RR^k$ be a linear map for some $1 \le k < d$. Then $\pi(S)\subseteq \RR^k$ is a semi-algebraic set of complexity 
	$b^{O(k(d-k))}$ and can be computed in $b^{O(k(d-k))}$ time.
\end{lemma}
\begin{proof}
	Using a linear change of coordinates, we can assume that $\pi$ maps $\RR^d$ onto the hyperplane $x_{k+1} = \cdots = x_d=0$. The claim now follows from Proposition~\ref{BasuBlockElimination} with $h=d-k$, $\ell=k$, $s,t \le b$.
\end{proof}

We will also need the following sightly technical result from algebraic geometry.
\begin{lemma}
  \label{goodProjectionCodimOne}
Let $S\subset\RR^d$ be a semi-algebraic set with $\dim(S)<d$. Then there is a $(d-1)$-dimensional subspace $H\subset\RR^d$ so that each fiber of the orthogonal projection $\pi\colon S\to H$ has finite cardinality.  
\end{lemma}
\begin{proof}
This follows from standard techniques in algebraic geometry, so we just briefly sketch the proof here. Let $V = \overline S$ be the Zariski closure of $S$; we have $\dim(V) = \dim(S) < d$. Let $\hat{V}$ be the projectivization of $V$ (i.e., the smallest projective variety in $\RR\mathbf{P}^d$ that contains the image of $V$ under the embedding $(x_1,\ldots,x_d)\mapsto[1{:}x_1{:}\cdots{:}x_d]$).\footnote{This embedding maps a point in $\RR^{d}$ to a point in the projective space $\RR\mathbf{P}^d$, which is formed by a line in $\RR^{d+1}$ that contains the origin.} Then the intersection of $\hat V$ with the ``hyperplane at infinity'' $H_{\infty} \coloneqq \{[0{:}x_1{:}\cdots{:}x_d] : (x_1,\dots,x_d) \in \RR^d \}$ has dimension at most $d-1$. But if $\ell\subset V$ is a line, then the direction of $\ell$ is contained in $\hat V \cap H$. 
	In particular, there is a direction $v$ so that no line pointing in direction $v$ is contained in $V$. Thus if $\ell$ is a line pointing in direction $v$, we have $|\ell \cap S|\leq |\ell\cap V|$ is finite. Let $H$ be the orthogonal compliment of $v$. 
\end{proof}

We now describe a procedure for computing a ``good'' projection for a given algebraic variety, which is a key step of our algorithm. Let  $V\subset\RR^d$ be a variety, let $H$ be a subspace of $\RR^d$, and let $\pi\colon\RR^d\to H$ denote the orthogonal projection to $H$. We say that $H$ \emph{compresses} $V$ if there is an irreducible component $V^\prime\subset V$ so that $\dim(\pi(V^\prime))<\dim(V^\prime)$.

\begin{lemma}
  \label{goodOrthogonalProjection}
Let $\PP \coloneqq \{P_1,\ldots,P_k\}\subset\RR[x_1,\ldots,x_d]$ be a set of $k \le d$ nonzero polynomials of degree at most $E$, and let $V \coloneqq Z(\PP).$ Then there is a subspace $H\subset\RR^d$ of dimension $\dim(V)$ that does not compress $V$. If $d$ and $E$ are considered to be constant then such a choice of subspace can be computed in time $O(1)$.
\end{lemma}
\begin{proof}
	Let $d^\prime  \coloneqq  \dim(V)$ and let $G$ be the set of $d^\prime$-dimensional subspaces of $\RR^d$. A $d'$-dimensional subspace $H$ can be specified by a set of $d'$ vectors in $\RR^d$ that span $H$, so $H$ can be viewed as a point in $\RR^{d\times d'}$ and $G \subset \RR^{d\times d'}$. 
	Define
\begin{equation}
G_{\operatorname{good}} \coloneqq \{H \in G \colon \forall x \in \RR^d,\ |V \cap (H^\perp+x)|<\infty\}.
\end{equation}
Note that the requirement $\forall x \in \RR^d,\ |V \cap (H^\perp+x)|<\infty$ is precisely the statement that the fibers of the orthogonal projection $\pi\colon V\to H$ have finite cardinality.

First, we claim that if $H\in G_{\operatorname{good}}$ then $H$ does not compress $V$. Indeed, if $H$ compresses $V$, then at least one fiber of the orthogonal projection $\pi\colon V\to H$ must have positive dimension and thus infinite cardinality. 

Next, we claim that $G_{\operatorname{good}}$ is non-empty. Define $S_0  \coloneqq  V$ and $H_0  \coloneqq  \RR^d$. For each index $i=1,\ldots,d-d^\prime$, apply Lemma \ref{goodProjectionCodimOne} to find a $(d-i)$-dimensional subspace $H_i\subset H_{i-1}$ so that each fiber of the orthogonal projection $\pi_i\colon S_{i-1}\to H_i$ has finite cardinality. We define $S_{i} \coloneqq \pi_i(S_{i-1})$ (this is a semi-algebraic set of dimension at most $d^\prime$) and repeat. We define $H  \coloneqq  H_{d-d^\prime}$ and let $\pi$ be the orthogonal projection $\RR^d\to H$. Since $\pi = \pi_{d-d^\prime}\circ\cdots\circ \pi_1$, it immediately follows that each fiber of $\pi\colon V\to H$ has finite cardinality.

Finally, we describe a procedure for computing an element of $G_{\operatorname{good}}$. If $H\in G\backslash G_{\operatorname{good}}$, then as we have discussed above, there exists $x\in\RR^d$ so that $|V \cap (H^\perp+x)|=\infty$. By Lemma \ref{milnorThomResult}, if $|V \cap (H^\perp+x)|<\infty$ then $| V \cap (H^\perp+x)|\leq E(2E-1)^{d-d^\prime-1}$.
In view of this observation, $G\backslash G_{\operatorname{good}}$ can be expressed as follows.
To simplify notation in what follows, set $F \coloneqq E(2E-1)^{d-d^\prime-1}+1$. 
\begin{equation}
\begin{split}
	G\backslash G_{\operatorname{good}} = \{ (v_1,\ldots,&v_{d^\prime})\in \RR^{d\times d^\prime}\colon   \exists x,y_1,\ldots,y_F\in\RR^d\ \textrm{such that}\\
&v_i\neq 0,\ i=1,\ldots,d^\prime;\\
&v_i\cdot v_j = 0,\ 1\leq i < j \leq d^\prime;\\
&y_i \neq y_j,\ 1\leq i < j \leq F;\\
&(y_i - x)\cdot v_j = 0,\ i=1,\ldots,F,\ j =1,\ldots,d^\prime;\\
&P_i(y_j) = 0,\ i=1,\ldots,k,\ j = 1,\ldots,F
\}.
\end{split}
\end{equation}
The first two conditions guarantee that the vectors $v_1,\ldots,v_{d^\prime}\in\RR^d$ are non-zero and pairwise orthogonal, and thus $H  \coloneqq  \operatorname{span}(v_1,\ldots,v_{d^\prime})$ is a $d^\prime$-dimensional subspace of $\RR^d$. The third condition guarantees that the points $y_1,\ldots,y_F\in\RR^d$ are distinct. The fourth and fifth condition specifies that each of the points $y_i$ is contained in $Z(\PP)\cap (H^\perp+x)$, and thus $| V \cap (H^\perp+x)|\geq F,$ which implies $|V \cap (H^\perp+x)|=\infty$. 

	We now apply Proposition~\ref{BasuBlockElimination} to compute a quantifier-free formula for $G\backslash G_{\operatorname{good}}$, and thus of $G_{\operatorname{good}}$. Finally, we use Proposition~\ref{arrangement} to compute a point $(v_1,\ldots,v_{d^\prime})\in G_{\operatorname{good}}$. We define $H \coloneqq \operatorname{span}(v_1,\ldots,v_{d^\prime})$.  By Propositions~\ref{BasuBlockElimination} and~\ref{arrangement}, if $E$ and $d$ are constants, it takes $O(1)$ time to compute $v$.
\end{proof}

\paragraph{One stage of the partitioning scheme.}
We now describe a single stage of our multi-level partitioning scheme.
  Let $\PP \coloneqq \{P_1,\ldots,P_k\}\subset\RR[x_1,\ldots,x_d]$ be a set of $k\le d$ nonzero polynomials of degree at most $E$, let $V \coloneqq Z(\PP)$, and let $d^\prime \coloneqq \dim(V)$. Let $\S$ be a multiset consisting of $n$ semi-algebraic sets in $\RR^d$, each of complexity at most $b$. For a given parameter $D>1$, the goal is to  construct a polynomial
  $P\in\RR[x_1,\ldots,x_d]$ of degree at most $D$ so that each connected component of $V\backslash Z(P)$ is crossed by $O(n/D)$ sets of $\S$. The algorithm performs the following steps to construct $P$:

  \begin{enumerate}
	  \item Using Lemma~\ref{goodOrthogonalProjection}, we compute a subspace $H$ of dimension $d'$ and an 
		  orthogonal projection $\pi: \RR^d \rightarrow H$ such that $\dim \pi(V') = \dim V'$ for every irreducible component of $V$.
	  \item For each set $S\in\S$, let $\bd S_V  \coloneqq  \bd_V (S\cap V)$ be the relative boundary of $S\cap V$ in  $V$. By Lemma~\ref{complexityOfRelativeBdry}, $\bd S_V$ has complexity $(bE)^{O(d)}$ and $\dim \bd S_V \le d'-1$.
		  Set $\bd \S_V \coloneqq \{ \bd S_V \mid S \in \S\}$.
	  \item For each set $S\in \S$, let $S^\downarrow  \coloneqq  \pi(\bd S_V)$ be the projection of $\bd S_V$ onto the $d'$-dimensional subspace $H$; by Lemma~\ref{lemm:project}, 
		  $S^\downarrow$ has complexity $(bE)^{O(d^3)}$ and $\dim S^\downarrow \le d'-1$.
	  \item  For  simplicity, we identify the subspace $H$ with $\RR^{d'}$. 
		  Using Theorem~\ref{thm:efficientlyComputePartition}, we construct a $d'$-variate partitioning polynomial $P^\downarrow$ of degree at most $D$ so that each cell of $\RR^{d'}\backslash Z(P^\downarrow)$ crosses 
		  at most $(bE)^{d^{O(1)}}n/D$ sets of $\S^\downarrow$. Let $\Omega^\downarrow$ be the set of resulting cells in $\RR^{d'}\backslash Z(P^\downarrow)$. 
		 The algorithm also computes a semi-algebraic representation of each cell $\omega^\downarrow \in \Omega^\downarrow$ as well as the subset $S_\omega^\downarrow$ of $\S^\downarrow$ intersecting $\omega^\downarrow$.
	 \item Let $P$ be the pullback of $P^\downarrow$ in $\RR^d$, i.e., $P(x) = P^\downarrow (\pi(x))$; $\deg P = \deg P^\downarrow \le D$. Using Proposition~\ref{arrangement}, we construct the cells of $V\setminus Z(P)$, i.e., a semi-algebraic representation of each cell and a reference point within each cell. 
		  Let $\Omega$ be the set of resulting cells. We note that the cells in $\Omega$ are stacked over the cells of $\Omega^\downarrow$ in the following sense: for each cell $\omega^\downarrow \in \Omega^\downarrow$, the connected components of $\pi^{-1}(\omega^\downarrow)\cap V$ are the cells of $\Omega$. Recall that the algorithm computes the subset $\S^\downarrow_{\omega^\downarrow} \subseteq \S^\downarrow$ of semi-algebraic sets that intersect $\omega^\downarrow$. For each cell $\omega\in \Omega$ such that $\pi(\omega) = \omega^\downarrow$, we compute the subset $\S_\omega \subseteq \{ S \mid S^\downarrow \in \S^\downarrow_{\omega^\downarrow}\}$ of elements that cross $\omega$.
  \end{enumerate}
  This completes the description of the algorithm for a single level of the partitioning scheme. We now prove the correctness of the algorithm and analyze its performance.

  \begin{lemma}
	  \begin{itemize}
		\item[(i)] $|\Omega| = O(E^{d-d'}D^{d'})$, and the complexity of each cell is $(ED)^{O(d^4)}$.
		\item[(ii)] Each cell of $\Omega$ is crossed by at most $(bE)^{d^{O(1)}} n/D$ elements of $\S$.
	  \end{itemize}
  \end{lemma}
  \begin{proof}
	  The bounds on the size of $\Omega$ follows from Proposition~\ref{BB-bound} and the complexity of each cell follows from Proposition~\ref{arrangement}, so it suffices to prove (ii).

	  Let $\omega$ be a cell in $\Omega$, and let $\omega^\downarrow  \coloneqq  \pi(\omega)$. Then $\omega^\downarrow$ is a cell in $\Omega^\downarrow$. Recall that $\omega^\downarrow$ intersects $(bE)^{d^{O(1)}}n/D$ sets of $\S^\downarrow$. Hence, $\omega \subset \pi^{-1} (\omega^\downarrow)$ intersects at most $(bE)^{d^{O(1)})}n/D$ sets of $\bd \S_V$. We now claim that if a set $S\in \S$ crosses $\omega$, then $\bd S_V$ intersects $\omega$.

 If $S\in\S$ crosses $\omega$, then there are points $p,q\in \omega$ with $p\in S$ and $q\not\in S$. Since $\omega$ is connected and semi-algebraic, there is a curve $\gamma\subset \omega$ with endpoints $p$ and $q$. Write $\gamma=(\gamma\cap S) \sqcup(\gamma\backslash S)$; both of these sets are non-empty, and since $\gamma$ is connected, $\overline{(\gamma\cap S)} \cap \overline{(\gamma\backslash S)}$ is non-empty. But $\overline{(\gamma\cap S)} \cap \overline{(\gamma\backslash S)}\subset \bd S_V\cap \omega$, so $\bd S_V$ intersects $\omega$. This completes the proof of the lemma.
  \end{proof}

  Assuming $E$ and $b$ are constants, Steps~(1)--(3) take $O(n)$ time, and Step (4) takes $O(n\poly(D) + e^{\poly(D)})$ expected time. As for Step~5, computing $\Omega$ takes $O(\poly(D))$ time, and it takes $O(\poly(D))$ time to determine whether a set $S\in\S$ crosses a cell $\omega\in\Omega$. Hence, Step~5 takes a total of 
  $O(n\poly(D))$ time. Putting everything together, we obtain the main result of this subsection.

\begin{lemma}
  \label{partitioningOnAVariety}
	Let $\PP  \coloneqq \{P_1,\ldots,P_k\}\subset\RR[x_1,\ldots,x_d]$ be a set of $k \le d$ nonzero polynomials of degree at most $E$, let $V \coloneqq Z(\PP),$ and let $d^\prime \coloneqq \dim(V)$. Let $\S$ be a multiset consisting of $n$ semi-algebraic sets in $\RR^d$, each of complexity at most $b$.

	Then for each $D>1$, there is a polynomial $P\in\RR[x_1,\ldots,x_d]$ of degree at most $D$ so that $V\cap Z(P)$ has dimension at most $d^\prime-1$; $V\backslash Z(P)$ is partitioned into a set $\Omega$ of $O(E^{d-d'}D^{d^\prime})$ connected semi-algebraic cells, each of complexity $(ED)^{O(d^4)}$, so that each cell of $\Omega$ is crossed by  at most $(bE)^{d^{O(1)}} n/D$ sets of $\S$. Assuming $E$ and $b$ are constants, the polynomial $P$, a semi-algebraic representation of the cells in $\Omega$,  and the elements of $S$ crossing each cell of $\Omega$ can be computed in $O(n\poly(D) + e^{\poly(D)})$ randomized expected time.
\end{lemma}

\paragraph{Computing the multi-level partition.}
We are now ready to describe the overall multi-level partitioning scheme that given $\S$ and a constant $\eps>0$ partitions $\RR^d$ into a family $\Omega_0, \ldots, \Omega_d$, where each $\Omega_i$ is a collection of (open) semi-algebraic cells of constant complexity so that any cell $\omega\in\Omega_i$ is crossed by at most $\tfrac{n}{4|\Omega_i|^{1/d-\eps}}$ elements of $\S$; see Theorem~\ref{multiLevelPartitionBdry} for a more precise statement. We construct this partition by 
invoking the single-stage algorithm described above at most $d$ times, as follows.

We select a sequence of constants $D_0,D_1,\ldots,D_{d}$, where $D_0$ depends only on $b,d,$ and $\eps$, and $D_i$ depends on $b,d,\eps,$ and $D_{i-1}$. Thus each of $D_0,\ldots,D_d$ ultimately depends only on $b,d,$ and $\eps$. 
For each index $i\le d$, there will also be a constant $C_i>0$ that depends only on $b,d,$ and $D_{i-1}$  (and not on $\eps$); $C_0$ depends only on $b$ and $d$. Then $D_i$ will be selected sufficiently large so that 
\begin{equation}
  \label{DiVsCi}
  C_i\leq \frac14 D_i^{\eps/4}.
\end{equation}

 At the beginning of the $i$th stage, for $0 \le i \le d$, we have a set of polynomials $\PP^{(i)}=\{P_0,\ldots,P_{i-1}\}$, where $\deg(P_j)\le D_j$ for all $0 \le j < i$, and we have constructed the  family $\Omega_0,\ldots,\Omega_{i-1}$ that partition 
 $\RR^d\backslash Z(\PP^{(i)})$, i.e., we have computed a semi-algebraic representation of each cell $\omega$ in $\Omega_j$ along with a reference point in $\omega$ and the subset $\S_\omega\subseteq\S$ of elements that cross $\omega$. The algorithm ensures that 
 $|\S_\omega| \le \tfrac{n}{4|\Omega_j|^{1/d-\eps}}$
 for every $\omega\in\Omega_j$.
 For $i=0$, $\PP^{(0)}=\emptyset$, we assume $Z(\PP^{(0)})$ to be $\RR^d$, and no semilagebraic cells have been constructed so far.

 In the $i$th stage, we perform the following steps. Let $V^{(i)}  \coloneqq   Z(\PP^{(i)})$. By construction, $\dim V^{(i)} \le d-i$. If $\dim V^{(i)}=0$, i.e., $V^{(i)}$ consists of a finite set of points, we set $\Omega_i \coloneqq V^{(i)}$ and the algorithm stops. So assume that $\dim V^{(i)}>0$. 
 We choose the constant $D_i$ as mentioned above and apply the above single-stage algorithm with 
 $\P^{(i)}$, $\S$, and $D_i$ as the input (note that $E=D_{i-1}$ in our setting). The algorithm returns a polynomial $P_i$ of degree at most $D_i$, so that $\dim(V^{(i)}\cap Z(P_i)) \le \dim V^{(i)}-1$, and a collection $\Omega_i$ of semi-algebraic cells that partition $V^{(i)}\backslash Z(P_i)$. For  each cell $\omega\in \Omega_i$, it also returns the subset $\S_\omega$ of elements in $\S$ that cross $\omega$. 
 We set $\PP^{(i+1)} \coloneqq \PP^{(i)}\cup \{P_i\}$ and proceed to the $(i+1)$st stage of the algorithm.

 If the algorithm terminates with $i < d$, we set $\Omega_{i+1}, \ldots, \Omega_d \coloneqq \emptyset$. This completes the description of the algorithm. We now analyze the performance of the algorithm.

\begin{theorem}
  \label{multiLevelPartitionBdry}
	Let $\S$ be a multiset consisting of $n$ semi-algebraic sets in $\RR^d$, each of complexity at most $b$, and
	let $\eps \in (0, min\{1/4, 1/d\})$ be a constant.
	Then there are collections $\Omega_0,\ldots,\Omega_{d}$ of semi-algebraic sets with the following properties.
  \begin{itemize}
	  \item[(i)] For each index $i$, each cell $\omega\in\Omega_i$ is a connected semi-algebraic set of complexity $O(1)$.
	  \item[(ii)] The cells partition $\RR^d$, in the sense that
    \begin{equation}\label{cellsPartitionRd}
      \RR^d=\bigsqcup_{i=0}^d \bigsqcup_{\omega\in\Omega_i}\omega.
    \end{equation}
  \item[(iii)] For each index $i$ and each $\omega\in\Omega_i$, let  $\S_\omega$ be the subset of sets in $\S$ that cross $\omega$. Then
	  $$|\S_\omega| \le \frac{n}{4|\Omega_i|^{1/d-\eps}}.$$
  \end{itemize}
  The sets in $\Omega_0,\ldots,\Omega_d$ can be computed in $O(n)$ expected time by a randomized algorithm. For each $i \le d$ and for each $\omega\in\Omega_i$, the algorithm specifies a semi-algebraic description of $\omega$, a ``reference point'' inside $\omega$, and the subset $\S_\omega$.
  The implied constants in the big-$O$ notation depend on $d$, $b$, and $\eps$.
\end{theorem}

\begin{proof}
	Part (i) follows from Lemma~\ref{partitioningOnAVariety} and the fact that a point is a connected semi-algebraic set. As for (ii), since $\dim V^{(i)}  \le d-i$, i.e., the dimension of $V^{(i)}$ goes down by at least $1$ at each stage, the algorithm terminates within $d+1$ stages and thus constructs a partition of $\RR^d$. 

	We now prove (iii). If $\dim V^{(i)}=0$, then no input set crosses any cell of $\Omega_i$, as each of which is a point, so (iii) obviously holds. Next, assume that $\dim V^{(i)} > 0$. In this case $\Omega_i \ne \emptyset$.
	By Lemma~\ref{partitioningOnAVariety}, there is a constant that depends only on $b, d$, and $D_{i-1}$ (but not on $\eps$) such that  $|\Omega_i|\leq C_iD_i^d$ and that each cell in $\Omega_i$ is crossed by at most $C_in/D_i$ sets of $\S$. Assuming $D_i$ is chosen sufficiently large so that \eqref{DiVsCi} holds, we have
\begin{equation}
  \label{DiVsOmega}
  \frac{1}{D_i^{d+\eps/4}}\leq \frac{1}{4C_iD_i^{d}} \leq \frac{1}{|\Omega_i|},
\end{equation}
and 
\begin{equation}
  \label{SVsOmega}
  \frac{C_in}{D_i}\leq \frac{n}{4D_i^{1-\eps/4}}.
\end{equation}
Combining \eqref{DiVsOmega} and \eqref{SVsOmega}, we obtain that
\begin{equation}
\frac{C_in}{D_i}\leq \frac{n}{4D_i^{1-\eps/4}}\leq \frac{n}{4|\Omega_i|^{\frac{1-\eps/4}{d+\eps/4}}} \leq \frac{n}{4|\Omega_i|^{1/d-\eps}},
\end{equation}
where in the last inequality we used the fact that $|\Omega_i|\geq 1$ and $\frac{1-\eps/4}{d+\eps/4}\geq\frac{1}{d}-\eps$ when $d\geq 1$ and $0<\eps<\min\{1/4, 1/d\}$. We conclude that if $\Omega_i$ is non-empty, then at most $\frac{n}{4|\Omega_i|^{1/d-\eps}}$ sets from $\S$ cross each cell of $\Omega_i$. The statement is also vacuously true if $\Omega_i$ is empty.

	Since $b, d, \eps, D_i$ are all constants, each application of Lemma~\ref{partitioningOnAVariety} takes $O(n)$ randomized expected time. Hence, the overall expected running time of the algorithm is also $O(n)$.
\end{proof}

\noindent\textbf{Remark.}
  In \cite{MP}, Matou\v{s}ek  and Pat\'akov\'a established a multi-level partitioning scheme in which the degrees $D_i$ of partitioning polynomials at each stage could be a function of $n$. Such a choice of $D_i$ will not work in the above proof, because the condition \eqref{DiVsCi} forces $D_i$ to be much larger than $D_{i-1}$, and thus the value of $D_i$  increases rapidly at each stage. In contrast,  since  Matou\v{s}ek  and Pat\'akov\'a are working with a point set, and the projection of a point set remains a point set, they do not face this problem.

\paragraph{Multi-level partitioning for semi-algebraic sets and points.}

We conclude this section by referring to the case where in addition to a set $\S$ of $n$ semi-algebraic sets in $\RR^d$,
we are also given a set $\Q$ of $m$ points in $\RR^d$, which we would like to partition simultaneously.
This scenario is not used directly by the applications presented in Section~\ref{sec:applications}, however, we believe this result is of independent interest.

In such a setting, we can adapt the multi-level polynomial partitioning scheme described above 
so that it also partitions the points.
Indeed, we can easily modify the single-stage algorithm to require that $O(m/D^{d^\prime})$ points from $\Q$ are contained in each cell, by applying the projection $\pi$ to the points in $\Q \cap V$ (where $V=Z(\PP)$ is the input variety), and then invoke Theorem~\ref{thm:efficientlyComputePartition} on both $\S^\downarrow$ and $\pi(\Q \cap V)$. Then Theorem~\ref{multiLevelPartitionBdry} would imply that we also have that, for each $i = 0, \ldots, d$ and $\omega \in \Omega_i$, at most $\frac{m}{4|\Omega_i|^{1-\eps}}$ points from $\Q$ are contained in $\omega$. This follows from a simple calculation similar to the one applied in the proof of Theorem~\ref{multiLevelPartitionBdry}. 
That is, each cell of $\Omega_i$ contains at most $C_im/D^d_i$ points of $\Q$, and from \eqref{DiVsCi} and the fact that $|\Omega_i|\leq C_iD^d_i$, we conclude that $|\Omega_i| < D_i^{d+\eps/4}$,
which implies
$$
\frac{C_im}{D_i^d} \le \frac{m}{4D_i^{d-\eps/4}} \le \frac{m}{4|\Omega_i|^{\frac{d-\eps/4}{d+\eps/4}}}
\le \frac{m}{4|\Omega_i|^{1-\eps}} ,
$$
as asserted.
We thus conclude the following:

\begin{cor}
  \label{cor:multiLevelPartitionBdry}
	Let $\S$ be a multiset consisting of $n$ semi-algebraic sets in $\RR^d$, each of complexity at most $b$, let $\Q$ be a multiset of $m$ points in $\RR^d$, and let $0<\eps<\min\{1/4, 1/d\}$ be a constant.
	Then there are collections $\Omega_0,\ldots,\Omega_{d}$ of semi-algebraic sets with the following properties.
  \begin{itemize}
  \item For each index $i$, each cell $\omega\in\Omega_i$ is a connected semi-algebraic set of complexity $O(1)$.
  \item For each index $i$ and each $\omega\in\Omega_i$, at most $\tfrac{n}{4|\Omega_i|^{1/d-\eps}}$ sets from $\S$ cross $\omega$, and at most $\tfrac{m}{4|\Omega_i|^{1-\eps}}$ points from $\Q$ are contained in $\omega$.
  \item The cells partition $\RR^d$, in the sense that
    \begin{equation}
      \RR^d=\bigsqcup_{i=0}^d \bigsqcup_{\omega\in\Omega_i}\omega.
    \end{equation}
  \end{itemize}
  The sets in $\Omega_0,\ldots,\Omega_d$ can be computed in $O(n + m)$ expected time by a randomized algorithm. 
	For  each $i$ and for every $\omega\in\Omega_i$, the algorithm returns a semi-algebraic representation of $\omega$, a ``reference point'' inside $\omega$, the subset of elements of $\S$ that cross $\omega$, and the subset of points of $\Q$ that are contained in $\omega$. 
\end{cor}

\section{Applications}
\label{sec:applications}

\subsection{Point-enclosure queries}
\label{sec:point-location}

Let $\S$ be a set of $n$ semi-algebraic sets in $\RR^d$, each of complexity at most $b$. Each set $S$ is assigned a weight $w(S)$. We assume that the weights belong to a semigroup, i.e., subtractions are not allowed, and that the semigroup operation can be performed in constant time.
We wish to preprocess $\S$ into a data structure so that the cumulative weight of the 
sets in $\S$ that contain a query point can be computed in $O(\log n)$ time; we refer to this query as \emph{point-enclosure} query.
Note that if the weight of each set is $1$ and the semi-group operation is Boolean $\vee$, then the point-enclosure query becomes a \emph{union-membership query}: determine whether the query point lies in $\bigcup\S$.
Another useful special case is when the weights are all~$1$ and the semi-group operation is integer addition, so that the query counts the number of sets of $\S$ covering it.

We follow a standard hierarchical partitioning scheme of space, e.g., as in \cite{CEGS-91,Ag-survey}, but use Theorem~\ref{multiLevelPartitionBdry} at each stage. 
Using this hierarchical partition, for a given constant $\delta>0$, we construct  a tree data structure $\T$ of depth $O(\log n)$ and size $O(n^{d+\delta})$,
and a query is answered by following a path in $\T$. 

\paragraph{The data structure.}
Set $\eps \coloneqq \delta/2d^2$  and we choose $n_0 \coloneqq n_0(\eps,d,b)$ to be a sufficiently large constant.
If $n\le n_0$ then $\T$ consists of a single node that stores $\S$ itself. So assume that $n > n_0$.  
Applying Theorem \ref{multiLevelPartitionBdry} to $\S$, we construct a collection $\Omega_0,\ldots,\Omega_d$, each a family of semi-algebraic sets, that together partition $\RR^d$. For each index $i$, each cell $\omega\in\Omega_i$ is crossed by at most $|\S|/(4|\Omega_i|^{1/d-\eps})$ sets of $\S$. For each index $i$ and each $\omega\in\Omega_i$, let $\S_{\omega}\subset S$ be the family of sets that cross $\omega$, and let $\S_\omega^* \subseteq \S$ be the family of sets that contain $\omega$; set $n_\omega  \coloneqq  |\S_\omega|$. 

The set $\S_\omega$ is specified by Theorem \ref{multiLevelPartitionBdry}. The set $\S^*_\omega$ consists precisely of those sets $S\in\S\backslash \S_{\omega}$ that contain the reference point of $\omega$; this set can be computed in $O(|\S|)$ time. We compute the weight $W_\omega   \coloneqq  w(\S^*_\omega)$ by applying the semigroup operation to the weights $\{w(S)\colon S\in\S^*_{\omega}\}$. This can be done in $O(n)$ time. 

We create the root $v$ of $\T$ and store the sets $\Omega_v  \coloneqq  \Omega_0 \cup \cdots\cup\Omega_d$ at $v$. 
For each $\omega\in\Omega_v$, 
we create a child $z_\omega$ and store $\omega$ and $W_\omega$ at $z_\omega$. We recursively construct the data structure for each $\S_\omega$ and attach it to $z_\omega$ as its subtree.
The total size of $v$ is $O(\sum_i |\Omega_i|)$, which is bounded by a constant $c_1 \coloneqq c_1(\eps,d,b)$.

Since $|\S_\omega| \le n/4$ and the recursion stops when the input size becomes at most $n_0$, the depth of $\T$ is $O(\log n)$. To bound the size of the data structure, let $\sigma(n)$ denote the maximum size of the data structure constructed on an input of size $n$. Then we obtain the following recurrence for $\sigma(n)$:
\begin{equation}
	\label{eq:space}
	\sigma (n) \le \left\{ \begin{array}{ll}
		n & \mbox{for $n \le n_0$,}\\[1mm]
		\displaystyle c_1 + \sum_{i=0}^d\sum_{\omega\in\Omega_i}\sigma(n_\omega)&\mbox{for $n>n_0$}.
	\end{array} \right. 
\end{equation}
We claim that the solution to the above recurrence is $\sigma(n) \le  n^{d+\delta}$.
The claim follows immediately for $n \le n_0$. So  using induction hypothesis and (\ref{eq:space}), we obtain
\begin{equation}
	\label{eq:size-recur}
\begin{split}
	\sigma(n) & \leq c_1 + \sum_{i=0}^d\sum_{\omega\in\Omega_i}\sigma(n_\omega)\\
& \leq c_1 + \sum_{i=0}^d \sum_{\omega\in\Omega_i} n_\omega^{d+\delta}\\
& \leq c_1 +  \sum_{i=0}^d \sum_{\omega\in\Omega_i}\left(\frac{n}{4|\Omega_i|^{1/d-\eps}}\right)^{d+\delta}\\
	& \leq c_1 + 4^{-d}n^{d+\delta}\sum_{i=0}^d \sum_{\omega\in\Omega_i}|\Omega_i|^{(d+\delta)(-1/d+\delta/2d^2)}\\ 
	& \leq c_1 + 4^{-d}n^{d+\delta}\sum_{i=0}^d \sum_{\omega\in\Omega_i}|\Omega_i|^{-1}\\ 
	& \leq \left(\frac{c_1}{n^{d+\delta}} + \frac{d+1}{4^d}\right ) n^{d+\delta}\\
& \leq n^{d+\delta}.
\end{split}
\end{equation}
The last inequality follows because $d \ge 1$, $n > n_0$, and $n_0$ is chosen sufficiently large (in particular $n_0 \ge 2c_1$).
An analogous argument shows that the expected preprocessing time is also $O(n^{d+\delta})$.

\paragraph{Query procedure.}
Given a query point $q\in \RR^d$, we compute the cumulative weight of the sets containing $q$ by traversing a path in the tree in a top-down manner: We start from the root and maintain a partial weight $W$, which is initially set to $0$. Suppose we are at  a node $v$. If $v$ is a leaf, we scan the list $\S_v$ stored at $v$. If $S\in \S_v$ contains $q$, then we add $w(S)$ to the partial weight $W$. If $v$ is an interior node, then we add the weight $W_v$ stored at $v$ to  $W$. Let $\Omega_v$ be the partition constructed at $v$. We find the cell $\omega\in\Omega_v$ containing $v$ and recursively visit the child $z_\omega$ of $v$.
The total query time is $O(\log n)$, where the constant of proportionality depends on $\eps$.
Putting everything together, we obtain the following:

\begin{theorem}
  \label{theo:point-location}
  Let $\S$ be a set of $n$ semi-algebraic sets in $\RR^d$, each of complexity at most~$b$ for some constant $b>0$, and let $w(S)$ be the weight of each set $S\in\S$ that belongs to a semigroup.  Let $\eps>0$ be a constant.
  Assuming that the semigroup operation can be performed in constant time, $\S$~can be preprocessed in $O(n^{d+\eps})$ randomized expected time into a data structure of size $O(n^{d+\eps})$  so that the cumulative weight of the sets that contain a query point can be computed in $O(\log n)$ time.
\end{theorem}

\subsection{Range searching}
\label{sec:range_search}
Next, we consider range searching with semi-algebraic sets: Let $P$ be a set of $n$ points in $\RR^d$. Each point $p\in P$ is assigned a weight $w(p)$ that belongs to a semigroup. Again we assume that the semigroup operation takes constant time.
We wish to preprocess $P$ so that, for a query range~$\gamma$, represented as a semi-algebraic set in $\RR^d$, the cumulative weight of $\gamma\cap P$ can be computed in~$O(\log n)$~time. Here we assume that the query ranges (semi-algebraic sets) are parameterized as described in Section~\ref{sec:semi-para}. That is, 
we have a fixed $b$-variate Boolean function $G$. A query range is represented as a point $x\in \XX=\RR^t$,
 for some $t\le \binom{b+d}{d}^b$, and the underlying semi-algebraic set is~$Z_{x,G}$. We refer to $t$ as the \emph{dimension of the query space}, and  to the range searching problem in which all query ranges are of the form $Z_{x,G}$ as \emph{$(G,t)$-semi-algebraic range searching}.

 For a point $p \in \RR^d$, let $S_p \subseteq \XX$ denote the set of semi-algebraic sets $Z_{x,G}$ that contain~$p$, i.e., $S_p=\{x\in\XX \mid p \in Z_{x,G}\}$. It can be checked that $S_p$ is a semi-algebraic set whose complexity depends only on $b, d$, and $G$.
 Let $\S \coloneqq \{S_p \mid p \in P\}$. For a query range $Z_{x,G}$, we now 
 wish to compute the cumulative weight of the sets in $\S$ that contain $x$. This can be done using Theorem~\ref{theo:point-location}. Putting everything together, we obtain the following:
\begin{theorem}
  \label{thm:range-search}
  Let $P$ be a set of $n$ points in $\RR^d$, let $w(p)$ be the weight of $p\in P$ that belongs to a 
  semigroup, and let $G$ be a fixed $b$-variate Boolean function for some constant $b>0$.
        Let $t\le \binom{b+d}{d}^b$ be the dimension of the query space.
  Assuming that the semigroup operation can be performed in constant time,
  $P$ can be preprocessed in $O(n^{t+\eps})$ randomized expected time into a data structure of size $O(n^{t+\eps})$,
        for any constant $\eps>0$, so that a $(G,t)$-semi-algebraic range query can be answered in $O(\log n)$ time.
\end{theorem}

\textbf{Remark.} If $G$ is the conjunction of a set of $b$ polynomial inequalities, then the size of the data structure can be significantly improved by using a multi-level data structure, with a slight increase in the query time to $O(\log^b n)$; see, e.g., \cite{Ag-survey}. Roughly speaking, the value of $t$ will now be the dimension of the parametric space of each polynomial defining the query semi-algebraic set, rather than the dimension of the parametric space of the entire semi-algebraic set (which is the conjunction of $b$ such polynomials). 

\subsection{Vertical ray shooting}
\label{sec:ray_shooting}

  Let $\S$ be a collection of $n$ semi-algebraic sets in $\RR^d$, each of complexity at most $b$. For a point $q\in \RR^d$, let $\rho_q  \coloneqq  \{ x_q+\lambda (0, \ldots, 0, 1) \mid \lambda >0\}$ be the (open) ray emanating from $q$ in the $(+x_d)$-direction. The vertical ray-shooting query asks for returning the first element of $S$ intersected by $\rho_q$, as one walks along $\rho_q$; if there is more than one such set, then any one of them can be returned arbitrarily. Note that we define $\rho_q$ to be an open set, so the intersection with the ray at point $q$ does not count. Our goal is to preprocess $\S$
  into a data structure so that a vertical ray-shooting query can be answered quickly.

\paragraph{The overall data structure.}
Our data structure for vertical ray shooting is similar to the one described in Section~\ref{sec:point-location}, 
except that we store an auxiliary data structure at each node of the tree to determine which of its children the query procedure should visit recursively.

Again we fix constants $\delta>0$, $\eps \le \delta/2d^2$,  and $n_0 \coloneqq n_0 (\eps, d, b)$.
If $n \le n_0$, the tree data structure $\T$ consists of a single node that stores $\S$.
So assume that $n > n_0$.
Using the multi-level partitioning algorithm described in Section~\ref{sec:multi_level}, 
we construct a collection $\Omega_0,\ldots,\Omega_d$  that satisfies
the properties mentioned in Theorem \ref{multiLevelPartitionBdry}. Let $\Omega  \coloneqq  \Omega_0 \cup \cdots \cup \Omega_d$.
For each cell $\omega\in\Omega$, it also computes the set $\S_{\omega}\subset S$ of elements that cross $\omega$.
We create the root node~$v$ of~$\T$ and a child~$z_\omega$ for each cell $\omega$. 
We store two auxiliary data structures $\DS_1(v)$ and $\DS_2(v)$ at~$v$, described below, each of which can be constructed in~$O(n^{d+\delta'})$ randomized expected time and requires $O(n^{d+\delta'})$ space, where $\delta' \le \delta/2$ is a constant. Given a query point $q\in\RR^d$, 
$\DS_1 (v), \DS_2(v)$ together determine, in $O(\log n)$ time, the cell $\omega\in\Omega$ that contains the first intersection point of~$\rho_q$ with a set of~$\S$.
We recursively construct the data structure for each subset $\S_\omega$ and attach it to $z_\omega$ as its subtree.

Let $\sigma(n)$ be the maximum size of the data structure on input of size $n$. Then we now obtain the following recurrence:

\begin{equation}
	\label{eq:ray-space}
	\sigma (n) \le \left\{ \begin{array}{ll}
		n & \mbox{for $n \le n_0$,}\\[1mm]
		\displaystyle c_1 n^{d+\delta'} + \sum_{i=0}^d\sum_{\omega\in\Omega_i}\sigma(n_\omega) & \mbox{for $n>n_0$,}
	\end{array} \right. 
\end{equation}
where $c_1 \coloneqq c_1(\delta,b,d)$ is a constant. By following the same analysis as in~(\ref{eq:size-recur}), we can verify that $\sigma(n) \le A n^{d+\delta}$ where $A \coloneqq A(\delta,b,d)$ is a sufficiently large constant.
Again, $\T$ can be constructed in $O(n^{d+\eps})$ randomized expected time.

For a query point $q\in\RR^d$, the first set hit by $\rho_q$ can be computed by traversing a root-to-leaf path in $\T$. Suppose we are at a node $v$. If $v$ is a leaf, then we scan the list $\S_v$ stored at $v$ and return the  element of $\S_v$ that $\rho_v$ first intersects.
Otherwise, we use the auxiliary data structures $\DS_1(v)$ and $\DS_2(v)$ to determine in $O(\log n)$ time the cell $\omega\in\Omega_v$ that
contains the first intersection point of $\rho_q$ with a set of $\S$. We recursively visit the child $z_\omega$ of~$v$. Since the depth of $\T$ is $O(\log n)$, the total query time is $O(\log^2 n)$.

This completes the description of the overall algorithm. What remains is to describe the auxiliary data structures $\DS_1, \DS_2$,  which are used to determine the cell $\omega\in\Omega$ that
contains the first intersection point of a vertical ray with a set of $\S$. We compute a refinement $\Omega^\nabla$ of $\Omega$ and then construct the two auxiliary data structures for each cell $\tau\in\Omega^\nabla$. We first describe how we compute $\Omega^\nabla$ and then describe the two data structures. Let $\PP$ be the set of at most $d+1$ polynomials that were computed  to construct the multi-level partition $\Omega$. Recall that the degree of each polynomial in $\PP$ is $D$, a constant that depends on $\delta, b$, and $d$.

\paragraph{Refinement of $\Omega$.}
We refine the cells of $\Omega$ using the so-called \emph{first-stage CAD (cylindrical algebraic decomposition)} of $\PP$; see, e.g., \cite[Chapter 5]{BPR} for a detailed overview of standard CAD.
That is, this is a simplified version of CAD, presented in \cite{AMS-13}. For the sake of completeness, we describe it here as well.

Roughly speaking, the first-stage CAD for $\PP$ is obtained by constructing a collection $\G \coloneqq \G(\PP)$ of polynomials in the variables $x_{1}, \ldots, x_{d-1}$ (denoted by $\mathrm{Elim}_{X_k}(\PP)$ in \cite{BPR}), whose zero sets contain the projection onto $\RR^{d-1}$  of all pairwise intersections $Z(P_i)\cap Z(P_j)$, $0 \le i < j \le d+1$, as well as 
set of points in $Z(P_i)$, for $0 \le i \le d+1$, of vertical tangency, self-intersection of zeros sets (roots with multiplicity), or a singularity of some other kind.
Having constructed $\G$, the first-stage CAD is obtained as the arrangement $\A(\PP \cup \G)$ in $\RR^d$, where the polynomials in $\G$ are now regarded as $d$-variate polynomials, that is, we lift them in the $x_d$-direction; the geometric interpretation of the lifting operation is to erect a ``vertical wall'' in $\RR^d$ over each zero set within $\RR^{d-1}$ of a $(d-1)$-variate polynomial from~$\G$, and the first-stage CAD is the arrangement formed by these vertical walls and zero sets of polynomials in $\PP$. It follows by construction that the cells of $\A(\PP \cup \G)$ are vertical stacks of ``cylindrical'' cells. In more detail, for each cell $\tau$ of $\A(\PP\cup\G)$, there is unique cell $\varphi$ of the $(d-1)$-dimensional arrangement $\A(\G)$ in $\RR^{d-1}$ such that one of the following two cases occur: 
(i)~$\tau =\{(x,\xi(x)\}\mid x\in \varphi\}$, where $\xi\colon \varphi\rightarrow\RR$ is a continuous 
semi-algebraic function (i.e., $\tau$ is the graph of $\xi$ over $\varphi$); or (ii)~$\tau =\{(x,t) \mid x\in \varphi, t\in (\xi_1(x), \xi_2(x))\}$, where $\xi_i$, $i=1,2$ is a continuous semi-algebraic function on~$\varphi$, the constant function $\varphi \rightarrow \{-\infty\}$, or the constant function $\varphi \rightarrow \{+\infty\}$, and $\xi_1(x)<\xi_2(x)$ for all $x\in\varphi$ (i.e., $\tau$ is a cylindrical cell over $\varphi$ bounded from below (resp.\ above) by the graph of $\xi_1$ (resp.\ $\xi_2$)), or it is unbounded.
As stated in \cite{AMS-13}, the total number of cells in $\A(\PP \cup \G)$ is $D^{O(d)}$,
and each of them is a semi-algebraic set defined by $D^{O(d^4)}$ polynomials of degree $D^{O(d^3)}$.

For a cell $\varphi$ of the $(d-1)$-dimensional arrangement $\A(\G)$, let $\St (\varphi)$ be the stack of cells of~$\A(\PP\cup\G)$ over $\varphi$, i.e., the set of cells that project to $\varphi$. 
We note that $\A(\PP)$ is a refinement of $\Omega$ and that $\A(\PP\cup\G)$ is a refinement of $\A(\P)$, therefore $\A(\PP\cup\G)$ is
a refinement of $\Omega$, denoted by $\Omega^\nabla$, as desired. 

\paragraph{Auxiliary data structures.}
As mentioned above, we construct $\DS_1, \DS_2$ on the cells of $\Omega^\nabla$ to quickly determine the desired cell of $\Omega$ that contains the first intersection point of a vertical ray with $\S$.
\smallskip

\textbf{\textit{The structure $\DS_1$.}}
Fix a cell $\tau$ of $\Omega^\nabla$.
$\DS_1$ is used to determine whether a query ray~$\rho_q$ whose source point $q$ lies in $\tau$ intersects any set of $\S$ inside $\tau$.

For each input set $S \in \S$ that crosses $\tau$, let $\column{S}$ be the set of points $x$ in $\RR^d$ such that 
the vertical ray $\rho_x$ intersects $S\cap \tau$, i.e., $\column{S}$ is the union of the (open) rays in the $(-x_d)$-direction emanating from the points of $S\cap\tau$. $\column{S}$ is a semi-algebraic set whose complexity depends only on $b, d$, and~$D$. 
Let $\column{\S}_\tau \coloneqq \{ \column{S} \mid S \in \S, S\cap\tau \ne \emptyset\}$. $\column{\S}_\tau$ can be computed in $O(n)$ time. 
Using Theorem~\ref{theo:point-location}, we process $\column{\S}$ into a data structure $\DS_1(\tau)$ 
of size $O(n^{d+\delta'})$, for a constant $\delta' \le \delta/2$, so that for a query point $q\in \RR^d$, we can determine in $O(\log n)$ time whether $q \in \bigcup\column{\S}$, i.e., whether $\rho_q$ intersects any set of $\S$ within $\tau$. We construct  $\DS_1(\tau)$ for all cells $\tau$ of $\Omega^\nabla$.  The total size of the data structure, summed over all cells of $\Omega^\nabla$
is $O(n^{d+\delta'})$, and it can be constructed in $O(n^{d+\delta'})$ randomized expected time; the constant hiding in the big-O notation depends on $\delta, D, b$, and $d$.
\smallskip

\textbf{\textit{The structure $\DS_2$.}}
Fix a cell $\tau\in\Omega^\nabla$.
$\DS_2$ is used to determine whether a line parallel to the $x_d$-axis intersects any set of $\S$ inside $\tau$.

For each input set $S \in \S$ that crosses~$\tau$, let $\shadow{S}_\tau$ be the projection of $S\cap\tau$ onto the hyperplane~$x_d=0$.  For a point $q \in \RR^d$, the vertical line (parallel to the $x_d$-axis) through~$q$ intersects $S$ inside $\tau$
if and only if $\shadow{q} \in \shadow{S}_\tau$ (where $\shadow{q}$ is the projection of $q$ onto the hyperplane $x_d=0$).
$\shadow{S}_\tau$ is a semi-algebraic set in $\RR^{d-1}$ whose complexity depends only on $b$ and $D$. Let $\shadow{\S}_\tau \coloneqq \{\shadow{S}_\tau \mid S \in \S, S\cap\tau \ne \emptyset\}$. $\shadow{\S}_\tau$ can be constructed in $O(n)$ time. Using Theorem~\ref{theo:point-location}, we process $\shadow{\S}_\tau$ into a data structure $\DS_2(\tau)$ of size $O(n^{d-1+\delta'})$ so that, for a query point $q\in\RR^d$, we can determine in $O(\log n)$ time whether 
$\shadow{q}\in\bigcup\shadow{\S}_\tau$, i.e., whether the vertical line through $q$ intersects any set of $\S$ inside $\tau$.
We construct $\DS_2(\tau)$  for all cells of $\Omega^\nabla$.
The total size of the data structure, summed over all cells of $\A(\{f\}\cup\G)$, is $O(n^{d-1+\delta'})$, 
and it can be constructed in $O(n^{d-1+\eps})$ randomized expected time; the constant hiding in the big-O notation again depends $\delta, D, b$, and $d$.
\smallskip

\textbf{\textit{Answering a query.}}
Given a query point $q\in\RR^d$, we determine the cell of $\A(\PP\cup\G)$ that contains the first intersection point of $\rho_q$ with a set of $\S$ as follows. 
First, we determine the cell $\tau$ of $\Omega^\nabla$ that contains the query point $q$ by brute force.
Using $\DS_1 (\tau)$, we determine in $O(\log n)$ time whether $\rho_q$ intersects $\S$ inside $\tau$. 
If the answer is yes, then $\tau$ is the desired cell. So assume that the answer is no. Let $\varphi$ be the cell containing $\shadow{q}$, in the $(d-1)$-dimensional arrangement $\A(\G)$. Let $\St(\varphi)  \coloneqq \langle \tau_1, \ldots, \tau_r\rangle$ be the stack of cells over $\varphi$, and let $\tau  \coloneqq  \tau_j$ for some $j \le r$. We visit the cells of $\St(\varphi)$ one by one in order, starting from $\tau_{j+1}$ until we find the first cell $\tau_k$ such that $\shadow{q} \in \bigcup \shadow{\S_{\tau_k}}$. 
Since $q$ lies below $\tau_k$, $\rho_q$ intersects $\S$ inside $\tau_k$ if and only if $\shadow{q} \in \bigcup\shadow(\S_{\tau_k})$. 
If there is no such cell, we conclude that $\rho_q$ does not intersect $\S$. Otherwise $\tau_k$ is the cell of $\A(\PP\cup \G)$ that contains the first intersection point of $\rho_q$ with a set of $\S$.
The total time spent by the query procedure is $O(\log  n)$, where once again the constant hiding in the big-O notation depends on $D, \delta, b$, and $d$.

Putting everything together we obtain the following:
\begin{theorem}
  \label{thm:vertical_ray_shooting}
  Let $\S$ be a collection of $n$ semi-algebraic sets in $\RR^d$, each of complexity at most
  $b$ for some constant $b>1$.
  $\S$ can be preprocessed, in $O(n^{d+\eps})$ randomized expected time,  into a data structure 
  of size $O(n^{d+\eps})$, for any constant $\eps > 0$, so that a vertical ray-shooting query can be answered in $O(\log^2 n)$ time.
\end{theorem}

\subsection{Cutting algebraic curves into pseudo-segments}
\label{sec:cuttings_curves}

Let $\C$ be a set of $n$ algebraic curves in $\RR^2$, each of degree at most $b$, such that no two of them share a component.
Sharir and Zahl \cite{SZ-17} presented a technique for cutting $\C$ into  $O(n^{3/2}\log^{O(1)}n)$
pseudo-segments by lifting $\C$ into a set $\hat{\C}$ of curves in $\RR^3$ as in~\cite{ESZ} and 
exploiting  Proposition~\ref{guthProp} for algebraic curves in $\RR^3$; see \cite[Theorem 1.1]{SZ-17}.%
\footnote{The analysis in~\cite{SZ-17} does not exploit the machinery of multi-level partitioning, and the zero set of the underlying partitioning polynomial is handled directly. 
}
Using Theorem~\ref{thm:efficientlyComputePartition}, we can convert their technique into an efficient algorithm  though the number of pseudo-segments into which $\C$ is cut increases slightly to $O(n^{3/2+\eps})$, for any constant $\eps>0$. We briefly sketch the algorithm and refer the reader to the original paper~\cite{SZ-17} for details.

  Following the technique in~\cite{ESZ} (see also the proof of Theorem~1.1 in \cite{SZ-17}), we lift each curve $C$ into  a curve $\hat{C}$ in $\RR^3$ of degree at most $b^2$. Let $\hat{\C}$ be the resulting set of curves in $\RR^3$; $\hat{\C}$ can be constructed in $O(n)$ time.
   A pair $C_1,C_2\in\C$ intersect more than once if and only if $\hat{C}_1, \hat{C_2}$ form a depth cycle in the $z$-direction, i.e., there are four points $(x_1,y_1,z_1), (x_1,y_1,z'_1), (x_2,y_2,z_2)$, and $(x_2,y_2,z'_2)$ such that $(x_1,y_1,z_1), (x_2,y_2,z_2)\in C_1$, $(x_1,y_1,z'_1), (x_2,y_2,z'_2)\in C_2$,  $z'_1>z_1$, and $z_2>z'_2$. As in~\cite{SZ-17}, we cut $\hat\C$  to eliminate all depth cycles; the projections of resulting arcs give the desired cutting of $\C$ into pseudo-segments.

   We choose two sufficiently large constants $n_0$ and $D$, depending on $b$ and $\eps$. If $n\le n_0$, we simply cut the curves in $\hat\C$ at points where their $xy$-projections intersect, i.e., cut at all pairwise intersection points of $\C$. If $n > n_0$, using Theorem~\ref{thm:efficientlyComputePartition}, we compute, in $O(n\poly(D)+e^{\poly(D)})$ randomized expected time, a partioning polynomial $P$ of degree at most $D$ such that each cell of $\RR^3\setminus Z(P)$ intersects at most $c_1 n/D^2$ curves of 
  $\hat\C$; $\RR^3\setminus Z(P)$ has at most $c_2 D^3$ cells. Here $c_1, c_2>0$ are constants that depend on $b$ but are independent of $D$ (and $\eps$). For each cell $\tau \in \RR^3\setminus Z(P)$, let $\hat\C_\tau$ be the set of curves that intersect $\tau$. The above algorithm computes $\hat\C_\tau$  for every cell $\tau$.

  Applying the procedure in~\cite{SZ-17} to $\hat\C$ and $P$, we cut the curves in $\hat\C$ into $O(D^2 n)$ pieces, resulting into a set $\hat\Gamma$ of arcs, so that all ``surviving'' depth cycles lie within a single cell of $\RR^3\setminus Z(P)$, i.e., if there are two arcs $\gamma, \gamma'\in\Gamma$ forming a depth cycle, then 
  both $\gamma$ and $\gamma'$ are contained in the same cell of $\RR^3\backslash Z(P)$ (cf. Lemma~A.2 in~\cite{SZ-17}; see also~\cite{AS-18})). For each cell $\tau\in\RR^3\setminus Z(P)$, we call the procedure recursively on the portions of $\hat\C_\tau$ that lie inside $\tau$. These recursive calls further refine the arcs of $\Gamma$ into smaller arcs. 

  When the above algorithm terminates, we return the $xy$-projections of the resulting 3D arcs as the desired pseudo-segments of $\C$. 

  The correctness of the algorithm follows from the proof of Theorem~1.1 in~\cite{SZ-17}. Let $T(n)$ denote the maximum expected running time of the algorithm for a set of $n$ curves, then we obtain the following recurrence:

\begin{equation}
	\label{eq:cut-time}
	T(n) \le \left\{ \begin{array}{ll}
		bn^2 & \mbox{for $n \le n_0$,}\\[1mm]
		\displaystyle c_2 D^3 T(c_1n/D^2) + c_3 (n\poly(D)+e^{\poly(D)}) & \mbox{for $n>n_0$,}
	\end{array} \right. 
\end{equation}
where $c_3 \coloneqq c_3(b,\eps)$ is a constant.
By induction on $n$, it can be verified that the solution of the above recurrence is $T(n) \le A n^{3/2+\eps}$, where $A$ is a sufficiently large constant that depends on $b$ and $\eps$. A similar argument shows that the algorithm cuts $\C$ into  $O(n^{3/2+\eps})$ pseudo-segments.
We thus conclude the following:
\begin{theorem}
  \label{cutting-curves}
  Let $\C$ be a set of $n$ algebraic curves in $\RR^2$, each of degree at most $b$, with no two sharing a common component, and let $\eps>0$ be a constant. Then $\C$ can be cut,
	in $O(n^{3/2 + \eps})$ randomized expected time, into 
  $O(n^{3/2 + \eps})$  Jordan arcs so that each pair of arcs intersect at most once. 
\end{theorem}

Following the same technique as in~\cite{SZ-17}, we can extend the above theorem to the following:
\begin{theorem}
  \label{cutting-arcs}
  Let $\Gamma$ be a set of $n$ Jordan arcs in $\RR^2$, each of which is contained in an algebraic curve of 
	degree at most $b$ and every pair of which have finite intersection, let $\chi$ be the number of 
	pairwise intersections between the arcs of $\Gamma$, and let $\eps>0$ be a constant. 
	Then $\Gamma$ can be cut,
	in $O(n^{1+\eps}+n^{1/2 - \eps}\chi^{1+\eps})$ randomized expected time, into 
  $O(n^{1+\eps}+n^{1/2 - \eps}\chi^{1+\eps})$   Jordan arcs so that each pair of arcs intersect at most once. 
\end{theorem}

Combining Theorem~\ref{cutting-curves} with the technique in~\cite{SZ-17} and applying the algorithm by Agarwal and Sharir~\cite{AS-05} for computing a set of marked faces in an arrangement of pseudo-segments, we obtain the following:

\begin{cor}
	\label{cor:marked-faces}
	Let $\C$ be a set of $n$ algebraic curves in $\RR^2$, each of degree at most $b$, with no two sharing a common component, let $\Q$ be a set of $m$ points in $\RR^2$, and let $\eps>0$ be a constant. The faces of $\A(\C)$ that contain a point of $\Q$ can be computed in $O(n^{3/2+\eps}+m^{1+\eps}+m^{2/3-\eps}n^{2/3+\eps})$ randomized expected time.
\end{cor}

\subsection{Eliminating depth cycles for triangles in ${\RR}^3$}

Aronov, Miller, and Sharir \cite{AMS-19} applied Proposition~\ref{guthProp} to lines in three dimensions in order to prove that $n$ pairwise disjoint non-vertical triangles in~${\RR}^3$ can be cut into $O(n^{3/2 + \eps})$ pieces that form a \emph{depth order}, for any $\eps > 0$ (the application of multi-level partitioning is not needed in the analysis of \cite{AMS-19}).
This extended and generalized an earlier result of Aronov and Sharir \cite{AS-18}  that $O(n^{3/2} \polylog{n})$ cuts are always sufficient to eliminate all depth cycles in a collection of $n$ pairwise disjoint non-vertical lines in ${\RR}^3$. Both bounds are close to worst-case optimal.

In \cite{AEZ-19}, three of the authors of the current paper explain how the argument of \cite{AS-18} can be turned into an efficient randomized algorithm which constructs the required cuts for a collection of lines in three dimensions; the expected running time of the algorithm and the number of cuts produced are near optimal, for the worst-case set of input lines.  In \cite{AEZ-19}, they proceed by replacing the partitioning guaranteed by Proposition~\ref{guthProp} by a not-quite polynomial decomposition described in that paper (in fact, this decomposition consists of a partitioning polynomial and a semi-algebraic set) and spelling out the details of how the remainder of the construction from \cite{AS-18} can be made fully effective using tools from computational real algebraic geometry, such as Proposition~\ref{pointLocationThm}.  (We note in passing that now we can use Theorem~\ref{efficientlyComputePartition}, rather than the decomposition from \cite{AEZ-19}, to accomplish the same task.)

Now we turn to the case of triangles.  As in \cite{AS-18}, the main reason the combinatorial bound of \cite{AMS-19} was not presented as an efficient and effective \emph{construction} was the non-constructive nature of Proposition~\ref{guthProp}.
The previous work of Aronov~\etal \cite{AEZ-19} addressed the case of lines in $3$-space, by integrating
the three-dimensional space decomposition they developed with the analysis in \cite{AS-18}.  In an earlier version of the paper, they also presented a brief description for the setting of triangles in $3$-space and the integration with the analysis in \cite{AMS-19}; for technical reasons the fact that the decomposition constructed in \cite{AEZ-19} is not a polynomial partition prevented its use as direct drop-in replacement for Proposition~\ref{guthProp}; see \cite{AMS-19} for a discussion of the somewhat subtle dependence of their analysis on the partitioning being polynomial.  So, as a result, the setting of triangles was dropped completely from the revised (and updated) version of~\cite{AEZ-19}, and it was left as an open problem whether the existing tools were sufficient to yield a corresponding efficient construction for cycle elimination for triangles.

In this section, equipped with Theorem~\ref{efficientlyComputePartition} we revisit the setting of triangles in $3$-space, and obtain an efficient algorithm to cut $n$ given triangles, as above, into $O(n^{3/2 + \eps})$ pieces that form a depth order.
Since Theorem~\ref{efficientlyComputePartition} produces an actual polynomial partitioning with properties similar to that in Proposition~\ref{guthProp}, integration with the analysis in \cite{AMS-19} is rather simple.

Let $\T$ be a set of $n$ pairwise disjoint non-vertical triangles
in ${\RR}^3$.
The analysis in \cite{AMS-19} recursively subdivides space, by applying a polynomial partitioning
of constant degree $D > 1$ for the $3n$ lines supporting the edges of the triangles in $\T$, where at every step in the recursion one draws curves on each triangle. These curves are eventually used to cut the triangles into pieces.
A close inspection of the analysis in \cite{AMS-19} shows that the total number of pieces is in fact proportional to the number of these ``cutting'' curves up to a polylogarithmic factor. Moreover, the time to produce the triangle pieces once these curves are given is also proportional to the number of curves, up to a polylogarithmic factor.
Therefore we focus on the construction of the cutting curves.

Let $f$ be a partitioning polynomial of degree $D > 1$ constructed at the current recursive step.
A cutting curve drawn on a triangle $\Delta \in \T$ is one of following three types:\footnote{In fact, the triangles $\Delta$ under consideration belong to a subset of $\T$ determined by their interaction with the cells in the partition; we omit these details in this discussion.}

\smallskip

\noindent{\bf (i)} \emph{Traces}: We draw $Z(f) \cap \Delta$ on $\Delta$, unless $\Delta \subset Z(f)$.
The underlying curve has degree at most $D$.
Regarding the complexity of this operations, we first need to determine whether $\Delta \subset Z(f)$, which can be done by computing \emph{resultants} \cite[Chapter~4]{BPR}.
If $\Delta \not\subset Z(f)$ we represent the curve  $\Delta \cap Z(f)$ by restricting $f$ to the plane of $\Delta$ and then computing the intersection of this curve with $\bd{\Delta}$. This latter operation exploits \emph{root representation} \cite[Chapter 10]{BPR}.
Overall, all these operations can be performed in $O(\poly(D))$ time per triangle \cite{BPR}.

\smallskip

\noindent{\bf (ii)} \emph{Critical shadows}: Let $S$ be the common zero set of $f$ and its $z$-derivative; this is the set of singular points or points of $z$-vertical tangency of $Z(f)$.
Let $H$ be the ``vertical curtain'' spanned by $S$, that is the union of all vertical lines that pass through points of $S$. As pointed out in \cite{AMS-19}, $H$ is a two-dimensional algebraic variety of degree $O(D^2)$. We next draw on $\Delta$ the curve $H \cap \Delta$.
Computing this curve is done by first constructing the solution set $S$ of the system $\{f = 0,\partial f/ \partial z =0\}$ in~$O(\poly(D))$ time, using properties of resultants and root representation \cite{BPR}, as above. Then the restriction to~$\Delta$ (recall that $\Delta$ is non-vertical) is done by projecting $S$ onto the plane containing $\Delta$, and then
intersecting this projection with $\bd{\Delta}$.
This can be again done in $O(\poly(D))$ time.

\smallskip

\noindent{\bf (iii)} \emph{Wall shadows}:
This step involves the construction of the \emph{vertical decomposition} of the zero set of
$k \coloneqq O(\log_D{n})$ polynomials $f$ (each of degree at most $D$) associated with all the ancestors of the current recursion node of our construction.\footnote{We refer the reader to \cite{AMS-19} for further details concerning the vertical decomposition of an arrangement of algebraic varieties.}
We denote by $F$ the product of these polynomials, and by $\V \coloneqq \V(Z(F))$ their vertical decomposition.
Given a triangle $\Delta$ under consideration, and an open three-dimensional cell $\nu$ of $\V$ meeting $\Delta$, we draw the one-dimensional boundary of $\nu \cap \Delta$ on~$\Delta$. If $\Delta \subset Z(F)$ is part of the floor or ceiling of such a cell $\nu$, we draw the boundary of the closure of $\nu$ on $\Delta$.

As pointed out in \cite{AMS-19}, the construction of $\V$ is based on producing a collection of
$O(k^2) = O(\log^2{n})$ curves, each of degree $O(D^2)$, where these curves are obtained by the
intersection of pairs of surfaces~$Z(f)$, and the loci of singular points and points with $z$-vertical tangencies on the individual surfaces.
Once these curves are produced, they are projected onto the plane, in order to produce a corresponding planar vertical decomposition (which is then lifted in the $z$-direction).
Omitting any further details, the construction of $\V$ can be completed in $O(\poly(D)\polylog{n})$ time.
The analysis in \cite{AMS-19} shows that on each triangle $\Delta$ we draw $O(D^4 \log^2{n})$ curves
overall. Note that by construction each cell of $\V$ is a vertical prism, whose floor and ceiling are portions of some $Z(f)$, and its edges are portions of either $z$-vertical segments or curves of degree at
most $O(D^2)$. Given this property, and applying similar considerations as in cases (i) and (ii), we conclude that in time $O(\poly(D)\polylog{n})$ we can produce all cutting curves of type (iii) on each triangle $\Delta$.

Recall that the analysis in \cite{AMS-19} recursively subdivides space using a partitioning polynomial $f$
of degree $D$, where, in each cell $\tau \in \RR^3 \setminus Z(f)$ the triangles $\Delta$ meeting $\tau$ are processed in turn. A crucial step of the analysis is that the recursive subproblem associated with the cell~$\tau$ only inherits the triangles that ``pierce''~$\tau$, that is, triangles whose boundary meets~$\tau$. The remaining triangles (referred to as ``slicing'') are not passed down the recursion, and we dispose of them as soon as we draw their corresponding cutting curves, as described above. 

It thus follows from the above discussion, the analysis in \cite{AMS-19}, and
Theorem~\ref{efficientlyComputePartition} that the expected running time to draw all cutting
curves over the triangles $\Delta \in \T$ satisfies the recurrence
\[
T(n) =  O(D^3) \cdot T(cn/D^2) + O( \poly(D) \cdot n \polylog{n}) ,
\]
where $c > 0$ is a constant.
This recursive relation is similar to the one shown in \cite{AMS-19} for
bounding the total number of curves drawn on the triangles in $\T$, 
and has a similar asymptotic solution.
Using induction on $n$, it is easy to verify that $T(n) = O(\poly(D)n^{3/2+\eps})$, once we pick a sufficiently large constant $D>0$; $\eps>0$ depends on~$D$ and can be made arbitrarily small by increasing~$D$, so we can rewrite the bound as $O(n^{3/2 + \eps})$. In summary, we have shown:

\begin{theorem}
  \label{depth_cycles}
  Let $\T$ be a collection of $n$ pairwise-disjoint non-vertical triangles %
  in~${\RR}^3$.
  Then for any constant $\eps > 0$, we can cut the triangles in $\T$ into $O(n^{3/2 + \eps})$ pieces,
  bounded by algebraic arcs of constant maximum degree $\delta = \delta(\eps)$,
  which form a depth order.
  These pieces can be produced in expected $O(n^{3/2+\eps})$ time.
\end{theorem}

\noindent\textbf{Remark.}
The authors in~\cite{AMS-19} sketch an extension of their result to the case in which the triangles in $\T$ are allowed to be vertical.  We believe that this extension can also be made algorithmic at no cost in the asymptotic running time or the number of pieces produced.  However, we omit the details from here and leave it as an exercise for the interested reader,

\section{Conclusion}
We presented an efficient randomized algorithm, based on singly exponential quantifier-elimination procedure~\cite{BPR}, that computes a partitioning polynomial for a given family of semi-algebraic sets in $\RR^d$. We also extended this algorithm to compute a partition of $\RR^d$ into (relatively open) semi-algebraic cells so that each cell is crossed by a few elements of $\S$, thereby extending the result by Matou\v{s}ek and Pat\'akov\'a for semi-algebraic sets. Finally, we presented a few applications of these results. We conclude by mentioning a few open problems:
\begin{itemize}
	\item[(i)] Is there a more direct algorithm for computing the partitioning polynomial (e.g. as in~\cite{AMS-13}) that does not rely on quantifier elimination and whose running time is polynomial in both $n$ and the degree of the partitioning polynomial? 
	\item[(ii)] The degree of the partitioning polynomial at each stage of the multi-level partitioning scheme increases rapidly, which precludes us choosing a non-constant value of the degree $D_i$ at each stage of the algorithm. Is there a different multi-level partitioning scheme in which the degree increases in a more controlled way and which allows us to choose $D_i$ to be $n^\delta$ for some small constant $\delta<1$? Such a scheme would replace the $n^\eps$ factors with $\polylog(n)$ factors in all the results in Section~\ref{sec:applications}.
	\item[(iii)] Theorem~\ref{thm:range-search} states that a semi-algebraic range query can be answered in $O(\log n)$ time using a data structure of $O(n^{t+\eps})$ size where $t$ is the dimension of the query space, which can be much larger than $d$. An open question is whether the size of the data structure be improved to $n^d$,  say, even when the query range is of the form $f(x) \ge 0$, where $f$ is a $d$-variate polynomial. For example, can a set $P$ of $n$ points in $\RR^2$ be preprocessed into a data structure of size $O(n^2)$ so that the number of points lying inside a query disk be computed in $O(\log  n)$ time? The best-known data structure for disk-counting query needs $O(n^3)$ space.
\end{itemize}

\paragraph{Acknowledgments.} The authors thank Saugata Basu and Micha Sharir for helpful discussions and two anonymous reviewers for their useful comments on the paper.

\end{document}